\documentclass[11pt,a4paper]{article}
\usepackage{amsmath}
\usepackage{amsfonts}
\usepackage{amssymb}
\usepackage{amsthm}
\usepackage{fullpage}
\usepackage{mathrsfs}
\usepackage{dsfont}
\usepackage{paralist}
\usepackage{natbib}
\usepackage{bm,bbm}
\usepackage{color}
\usepackage[colorlinks=true,linkcolor=blue,citecolor=blue,pdfborder={0 0 0}]{hyperref}
\usepackage{graphicx}
\usepackage{bbm}
\usepackage{mathabx} 
\usepackage{booktabs}

\newcount\Comments  
\newcommand{\kibitz}[2]{\ifnum\Comments=1\textcolor{#1}{#2}\fi}

\renewcommand{\em}{\it}

\newtheoremstyle{normal}
{2ex}               
{3ex}               
{}                  
{}                  
{\bfseries} 
{}                  
{2pt}   
{\thmname{#1}\thmnumber{ #2.} \thmnote{(#3)}}

\newtheoremstyle{italic}
{2ex}
{3ex}
{\itshape}
{}
{\bfseries} 
{}
{2pt}
{\thmname{#1}\thmnumber{ #2.} \thmnote{(#3)}}

\theoremstyle{normal}
\newtheorem{definition}{Definition}[section]
\newtheorem{remark}[definition]{Remark}

\theoremstyle{italic}
\newtheorem{theorem}[definition]{Theorem}
\newtheorem{lemma}[definition]{Lemma}
\newtheorem{prop}[definition]{Proposition}

\newcommand\N{\mathbb{N}}

\newcommand\R{\mathbb{R}}

\newcommand\eps{\varepsilon}
\newcommand{\vect}{\bm }  
\newcommand\Prob{\mathbb{P}}    
\newcommand\Exp{\mathbb{E}}     
\newcommand\ind{\mathds{1}}     

\newcommand\Hc{\mathcal{H}}

\newcommand\Sc{\mathcal{S}}

\newcommand\Ab{\mathbb{A}}
\newcommand\Bb{\mathbb{B}}
\newcommand\Cb{\mathbb{C}}
\newcommand\Db{\mathbb{D}}

\newcommand\Lb{\mathbb{L}}

\newcommand\weak{\rightsquigarrow}
\newcommand{\ip}[1]{\lfloor #1 \rfloor}
\newcommand{\p}{\overset{\Prob}{\to}}

\allowdisplaybreaks[1]

\begin{document}

\title{Detecting breaks in the dependence of multivariate extreme-value distributions}

\author{Axel B\"ucher\,\footnote{Ruhr-Universit\"at Bochum,
Fakult\"at f\"ur Mathematik, Universit\"atsstr. 150, 44780 Bochum, Germany.  {Tel:} +49/234\ 32\ 23286, 
{E-mail:} \texttt{axel.buecher@rub.de}. Corresponding Author.}
,~
Paul Kinsvater\,\footnote{Technische Universit\"at Dortmund, Fakult\"at Statistik, Vogelpothsweg 87,
44221 Dortmund, Germany. {E-mail:} \texttt{kinsvater@statistik.tu-dortmund.de}
} 
~and 
Ivan Kojadinovic\,\footnote{Universit\'e de Pau et des Pays de l'Adour, 
Laboratoire de math\'ematiques et applications, 
UMR CNRS 5142, B.P. 1155, 64013 Pau Cedex, France.
{E-mail:} \texttt{ivan.kojadinovic@univ-pau.fr}
}
}

\maketitle

\begin{abstract}
In environmental sciences, it is often of interest to assess whether the dependence between extreme measurements has changed during the observation period. The aim of this work is to propose a statistical test that is particularly sensitive to such changes. The resulting procedure is also extended to allow the detection of changes in the extreme-value dependence under the presence of known breaks in the marginal distributions. Simulations are carried out to study the finite-sample behavior of both versions of the proposed test. Illustrations on hydrological data sets conclude the work.
\end{abstract}






\section{Introduction} \label{sec:intro}
\def\theequation{1.\arabic{equation}}
\setcounter{equation}{0}

The study of extremes is of importance in many environmental applications. Prominent examples are the analysis of floods (\citealp{HosWal05}), heavy rainfalls (\citealp{CooNycNav07}) and extreme temperatures (\citealp{KatBro92}). Many of these problems are intrinsically multivariate; for instance, the severity of a flood depends not only on its peak flow, which is considered in many univariate flood studies, but also on its volume and its duration (\citealp{YueQuaBobLegBru99}). Catastrophic flood events typically occur when more than one of these variables is taking a high value and therefore, the analysis of the joint behavior is of key importance. In a river system, where flood data are collected from a number of stations, inference at a specific location can be greatly improved by incorporating observations from neighboring stations (\citealp{HosWal05}). Similarly, extreme temperatures are commonly studied at several stations simultaneously. 

In most of these environmental applications, it is common practice to assess the extreme observations by (modifications of) the \textit{annual block maxima method}, popularized in the classical monograph by \cite{Gum58}. Univariate observations collected on, say, a daily basis are aggregated by taking maxima over a longer time period, usually a year for stationarity considerations, resulting in a sample of maxima capturing most of the extreme outcomes. Thanks to the extremal types theorem, univariate block maxima are approximately distributed according to the (parametric) generalized extreme distribution, while the dependence within a vector of block maxima can be described approximately by a (nonparametric) extreme-value copula (\citealp{DehFer06}). The resulting joint distributions are called multivariate extreme-value distributions and serve as a widely accepted model for multivariate block maxima.

In statistical applications, it is common practice to assume that the time series of block maxima is stationary, temporarily independent and that its stationary distribution is exactly of the extreme-value type. Except for stationarity, the above assumptions are theoretically justified by the results of \cite{BucSeg14} if the focus is solely on the dependence structure. In the current paper, we address the issue of stationarity by developing tests for change-point detection within the multivariate contemporary distribution. 

More precisely, assuming that we observe a sample of independent multivariate observations $\vect X_1,\dots,\vect X_n$, where each $\vect X_i$ follows a multivariate extreme-value distribution whose c.d.f.\ is denoted $H^{(i)}$, we develop a test for the hypothesis
\begin{align}\label{eq:H0}
\Hc_0:\ H^{(1)}=\ldots=H^{(n)}
\end{align}
against alternatives involving the non-constancy of the c.d.f.s.  Since the univariate version of this problem has been treated, for instance, in \cite{JarRen08} using results from Chapter~1 of \citet{CsoHor97}, we will be particularly interested in the multivariate setting throughout this paper.  

Outside of the extreme-value framework, there is a huge amount of literature on detecting deviations from $\Hc_0$, see the classical monograph of \cite{CsoHor97} or \cite{AueHor13} for a recent review. 
It is useful to note that, by Sklar's theorem \citep{Skl59}, we can rewrite $\Hc_0$ as $\Hc_{0,m}\cap \Hc_{0,c}$, where
\begin{align}
 \Hc_{0,m}:&\ H^{(1)},\ldots,H^{(n)}\ \text{have time-homogeneous marginal c.d.f.s}, \label{eq:H0m} \\
 \Hc_{0,c}:&\ H^{(1)},\ldots,H^{(n)}\ \text{have the same copula (i.e., dependence).} \label{eq:H0c}
\end{align}
Roughly speaking, tests for $\Hc_0$ can be divided into two groups: tests that are powerful against deviations from $\Hc_{0,m}$ but which are rather insensitive to deviations from $\Hc_{0,c}$, and vice versa, see \cite{BucKojRohSeg14} for a discussion. In the present setting of observing data from a multivariate extreme-value distribution, the tests considered for instance in \cite{JarRen08} can be used to detect deviations from $\Hc_{0,m}$, whence it will be our main concern to design a test that is particularly sensitive to deviations from $\Hc_{0,c}$ when $C$ is known to be an extreme-value copula. Note that none of the existing tests for changes in the copula (see, e.g., \citealp{BucKojRohSeg14,KojQueRoh15}, among others) incorporates the latter information, whence an improvement in the power properties seems possible. In fact, our simulation study partially reported in Section~\ref{sec:sim} suggests that our proposed testing procedure is indeed superior to existing methods.

The test tailored to deal with extreme-value dependence that we propose is however affected by the same limitations as the aforementioned more general testing procedures: it can be used to reject $\Hc_{0,c}$ only if $\Hc_{0,m}$ holds. In some situations, although there are reasons to consider that $\Hc_{0,m}$ is not true, it is still of interest to assess whether $\Hc_{0,c}$ holds or not. For instance, in the hydrological applications to be presented in Section~\ref{sec:illustration}, the construction of dams during the observation period suggests that there might be potential breaks in the marginal distributions of extreme peak flows or volumes, while it is still of interest to assess whether the dependence between the variables of interest has changed or not. A second contribution of this work is thus to propose an extension of the studied testing procedure that can be used to detect deviations from $\Hc_{0,c}$ under certain types of simple departures from $\Hc_{0,m}$. 

The remaining part of the paper is organized as follows: The second section is devoted to mathematical preliminaries about extreme-value copulas. In Section~\ref{sec:testdim2}, a first version of the testing procedure is presented along with theoretical results establishing its asymptotic validity under $\Hc_0$. An extension of the studied test that can be used to detect change in the extreme-value dependence under known marginal breaks is proposed in Section~\ref{sec:testdim2break} along with generalizations of the theoretical results of Section~\ref{sec:testdim2}. Section~\ref{sec:sim} partially reports the results of extensive Monte Carlo experiments. Two illustrations on hydrological data sets are finally presented in Section~\ref{sec:illustration}. All proofs are deferred to a sequence of appendices.

\section{Notation and mathematical preliminaries } \label{sec:not}
\def\theequation{2.\arabic{equation}}
\setcounter{equation}{0}

In this section, we set the notation and gather some mathematical preliminaries needed throughout the paper.
Let $\vect{X}=\left(X_1,\ldots,X_d\right)$ be a $d$-dimensional random vector with continuous marginal c.d.f.s $F_1,\ldots,F_d$. By Sklar's theorem \citep{Skl59}, the joint c.d.f.\ $H$ of $\vect{X}$ can be uniquely decomposed as
\begin{align}\label{eq:sklar}
H(\vect{x})=C\left\{ F_1(x_1),\ldots,F_d(x_d)\right \} , \quad \vect{x}=(x_1,\ldots,x_d)\in\mathbb{R}^d,
\end{align}
where the so-called copula $C:[0,1]^d\rightarrow[0,1]$ is the joint c.d.f.\ of $\vect{U}=(U_1,\ldots,U_d)$ with $U_j=F_j(X_j)$.  

Throughout the paper, we will assume that the copula $C$ in~\eqref{eq:sklar} is an extreme-value copula. As a consequence of the results in \cite{Pic81}, these copulas can be parametrized by a function $A:\Sc_{d-1}\rightarrow[1/d,1]$ on the $(d-1)$-dimensional unit simplex $\Sc_{d-1}=\{\vect{t}=(t_2,\ldots,t_d)\in[0,1]^{d-1}:\ t_2+\ldots+t_d\leq1\}$, usually referred to as the \textit{Pickands dependence function}. More precisely, we have
\begin{align}\label{eq:evc}
C(\vect{u})=\exp\left\{\left(\sum_{j=1}^d\log u_j\right) A\left(\frac{\log u_2}{\sum_{j=1}^d\log u_j},\ldots,\frac{\log u_d}{\sum_{j=1}^d\log u_j}\right)\right\}
\end{align}
for any $\vect u \in(0,1]^d\setminus\{(1,\dots,1)\}$. If relation~\eqref{eq:evc} is met, then $A$ is necessarily convex and satisfies the boundary condition $\max\{1- \sum_{j=2}^d t_j, t_2, \dots, t_d\} \le A(\vect t) \le 1$ for all $\mathbf t = (t_2,\dots,t_d) \in \Sc_{d-1}$. The latter two conditions are, however, not sufficient to characterize the class of Pickands dependence functions unless $d=2$, see, e.g., \cite{BeiGoeSegTeu04} for a counterexample and \cite{Res13} for recent results concerning the case $d > 2$.

Assuming to observe a sample $\vect X_i$, $i=1,\dots,n$, of serially independent random vectors such that $\vect X_i$ has copula $C^{(i)}$ and corresponding Pickands dependence function $A^{(i)}$, we can use the representation in~\eqref{eq:evc} to rewrite $\Hc_{0,c}$ in~\eqref{eq:H0c} equivalently as  
\begin{equation}
\label{eq:H0A}
\Hc_{0,c} : \ A^{(1)} = \dots = A^{(n)} = A.
\end{equation}
The test statistics in the subsequent sections will be particularly designed for detecting deviations from this hypothesis.

In the rest of the paper, for notational convenience, we will work mostly in dimension $d=2$. Furthermore, given a set $T$, $(\ell^\infty(T), \| \cdot \|_\infty)$ will denote the space of real-valued, bounded functions on $T$ equipped with the supremum norm $\| \cdot\|_\infty$. The arrow $\weak$ denotes weak convergence in the sense of Hoffmann-J\o rgensen, see \cite{VanWel96}.


\section{Test statistics for $d=2$ under stationarity of the marginals} \label{sec:testdim2}
\def\theequation{3.\arabic{equation}}
\setcounter{equation}{0}

We begin by restricting ourselves to the case of dimension $d=2$ and assume to observe an independent sample $(X_i,Y_i)$, $i=1,\dots,n$, such that $(X_i,Y_i)$ has unknown c.d.f.\ $H^{(i)}$, copula~$C^{(i)}$ and continuous margins $F^{(i)}$ and $G^{(i)}$, respectively, where each $C^{(i)}$ is assumed to be an extreme-value copula with Pickands dependence function $A^{(i)}:[0,1] \to [1/2,1]$.

We aim at developing tests for change-point detection that are consistent against deviations form $\Hc_{0,c}$ in~\eqref{eq:H0A} provided $\Hc_{0,m}$ in~\eqref{eq:H0m} holds. 

Our test statistic will be of the CUSUM-type and will be based on the rank-based estimator of $A$ proposed by \cite{Fer13}. The underlying idea of that estimator is as follows: if $(U,V)$ is distributed according to some extreme-value copula $C$ with Pickands dependence function $A$, then
\[
A(t) = \frac{\Exp\left\{ \max\left (U^{1/(1-t)}, V^{1/t} \right) \right\} } {1- \Exp\left\{ \max\left (U^{1/(1-t)}, V^{1/t} \right) \right\} } , \quad t \in [0,1],
\]
with the convention that $u^{1/0} = 0$ for any $u\in (0,1)$. 
The function 
\begin{equation}
\label{eq:St}
S(t)  = \Exp\left \{ \max\left (U^{1/(1-t)}, V^{1/t} \right) \right\}, \quad t \in [0,1],
\end{equation}
is simply an expected value, whence defining an estimator for $S$ under the null hypothesis $\Hc_0 = \Hc_{0,m} \cap \Hc_{0,c}$ in~\eqref{eq:H0} is straightforward once we have estimated the unobservable sample $(U_i,V_i) = (F(X_i),G(Y_i))$, $i=1,\dots,n$. To do so, it is common to compute the scaled ranks
\begin{equation}
\label{eq:pseudo1n}
\hat{U}_{1:n,i}=\frac{1}{n+1}\sum_{j=1}^n \ind \left(X_j\leq X_i\right),\quad 
\hat{V}_{1:n,i}=\frac{1}{n+1}\sum_{j=1}^n \ind \left(Y_j\leq Y_i\right), \quad i=1,\dots,n,
\end{equation}
frequently referred to as {\em pseudo-observations} from the unknown copula~$C$. Then, a natural estimator of $S$ is simply
\[
\hat S_{1:n}(t) = \frac{1}{n} \sum_{i=1}^n \max \left( \hat U_{1:n,i}^{1/(1-t)}, \hat V_{1:n,i}^{1/t} \right), \quad t \in [0,1].
\]
The corresponding estimator of $A$, namely 
\[
\hat A_{1:n}(t) = \hat S_{1:n}(t) / \{1 - \hat S_{1:n}(t)\}, \quad t \in [0,1],
\]
was shown to be consistent and was investigated empirically in \cite{Fer13}. As a by-product of our work, we establish the asymptotic distribution of the process $\sqrt n ( \hat A_{1:n} - A)$ in the following proposition proved in Appendix~\ref{sec:proofs}. To the best of our knowledge, this result is new and might be of independent interest, for instance for the construction of uniform confidence bands. 

\begin{prop}\label{prop:ferreira}
Suppose that $(X_i,Y_i)$, $i=1,\dots,n$, are i.i.d.\ from a bivariate distribution with continuous marginal c.d.f.s and extreme-value copula $C$ whose Pickands dependence function $A$ is continuously differentiable on $(0,1)$. Then, in the normed space $(\ell^\infty([0,1]), \| \cdot \|_\infty)$, 
$
\sqrt n ( \hat A_{1:n} - A) \weak \Lb_C,
$ 
where
\[
\Lb_C(t) = - \{1+A(t) \}^2  \int_0^1 \Cb_C(0,1,y^{1-t},y^t) \, dy
\]
and $\Cb_C(0,1,\cdot,\cdot)$ is the weak limit of the empirical copula process \citep[see, e.g.,][]{Seg12} as defined in Proposition~\ref{prop:weakdim2} below.
\end{prop}

In order to define a CUSUM-type procedure for testing $\Hc_0$ in~\eqref{eq:H0}, we consider the following subsample analogues of $\hat S_{1:n}$ based on subsequences $(X_k, Y_k) ,\dots, (X_\ell, Y_\ell)$, $1\le k \le \ell \le n$, namely
\begin{equation}
\label{eq:hatSkl}
\hat{S}_{k:\ell}(t)
=
\frac{1}{\ell-k+1}\sum_{i=k}^\ell \max\left(\hat{U}_{k:l,i}^{1/(1-t)},\hat{V}_{k:l,i}^{1/t}\right), \quad t \in [0,1],
\end{equation}
where
\begin{equation}
\label{eq:pseudoobs}
\hat{U}_{k:\ell,i}=\frac{1}{\ell-k+2}\sum_{j=k}^\ell \ind \left(X_j\leq X_i\right),\quad 
\hat{V}_{k:\ell,i}=\frac{1}{\ell-k+2}\sum_{j=k}^\ell\ind \left(Y_j\leq Y_i\right),
\end{equation}
with the convention that $\hat{S}_{k:\ell} = 0$ for $k > \ell$. The corresponding subsample estimators of $A$ are then simply 
\begin{equation}
\label{eq:hatAkl}
\hat A_{k:\ell} (t) = \hat S_{k:\ell}(t) / \{1 - \hat S_{k:\ell}(t)\}, \quad t \in [0,1].
\end{equation}

Under $\Hc_0$, the difference between $\hat A_{1:k}$ and $\hat A_{k+1:n}$ should be small for any $k=1, \dots, n-1$, which suggests to base test statistics for $\Hc_0$ on the process
\begin{align} \label{eq:Dn}
\Db_n(s,t)
=
\frac{\ip{ns}(n-\ip{ns})}{n^{3/2}}\left\{\hat{A}_{1:\ip{ns}}(t)-\hat{A}_{\ip{ns}+1:n}(t)\right\},
\end{align}
for $(s,t) \in [0,1]^2$. Typical test statistics would be given by Kolmogorov--Smirnov or $L^2$-type functionals of $\Db_n$. Throughout this paper, we focus on the hybrid version
\begin{align}\label{eq:SnA}
S_{n,A} = \max_{1\leq k<n}\int_{[0,1]}\left\{ \Db_n(k/n,t)\right \}^2\, d\mu(t),
\end{align}
where $\mu$ denotes some finite measure on $[0,1]$. In the finite-sample experiments of Section~\ref{sec:sim}, we use $\mu = T^{-1}\sum_{t\in\Gamma}\delta_{t}$ for some finite grid $\Gamma=\{t_1,\ldots,t_T\}\subset[0,1]$, where $\delta_t$ is the Dirac mass at $t$. A related two-sample statistic for detecting breaks at some pre-specified time point $k^\star$ is simply given by
\begin{align} \label{eq:SnAk}
S_{n,A}(k^\star) 
= 
\int_{[0,1]}\left\{ \Db_n(k^\star/n,t)\right \}^2\, d\mu(t).
\end{align}
Note that the aforementioned test statistics do not incorporate the information that the marginal distributions are (or should be close to) generalized extreme-value (GEV) distributions. There are several reasons for ignoring that information: First of all, since we are only interested in the dependence, it seems natural to ignore any marginal information. This is theoretically justified by results in \cite{GenSeg10}, where it is shown that rank-based estimation of copulas may be substantially more efficient than estimation based on even exact knowledge of the marginal distributions.
Second, even if we opt for a semiparametric estimation of the dependence structure based on a parametric estimation of the marginal distributions, we could expect the estimation of the copula to be quite negatively affected by the fact that estimators of GEV parameters have frequently a large variance. Last but not least, extreme-value copulas are also of interest outside of the genuine extreme-value framework, where margins are not necessarily of the GEV-type.

The following proposition, proved in Appendix~\ref{sec:proofs}, establishes weak convergence of the key process $\Db_n$ in~\eqref{eq:Dn} under $\Hc_0$ in~\eqref{eq:H0}. It is essential for deriving the weak limit of the preceding test statistics.

\begin{prop} \label{prop:weakdim2}
Suppose that $\Hc_0$ holds and that $A$ is continuously differentiable on $(0,1)$. Then, in the normed space $(\ell^\infty([0,1]^2), \| \cdot \|_\infty)$, $\Db_n \weak \Db_C$, where
\begin{equation}
\label{eq:DbC}
\Db_C(s,t)=\{1+A(t)\}^2 \int_0^1s \Cb_C(s,1,y^{1-t},y^t) - (1-s) \Cb_C(0,s,y^{1-t},y^t) \, dy.
\end{equation}
Here, $\Cb_C$ denotes a centered Gaussian process on $\Delta \times [0,1]^2$, with $\Delta = \{(s,s') \in [0,1]^2: s\le s'\}$,  defined through
\begin{multline*}
\Cb_C(s,s',u,v)= \{ \Bb_C(s',u,v) - \Bb_C(s,u,v)  \}  - \dot C_1(u,v) \{ \Bb_C(s',u,1) - \Bb_C(s,u,1)  \}   \\
- \dot C_2(u,v) \{ \Bb_C(s',1,v) - \Bb_C(s,1,v)  \} 
\end{multline*}
if $(u,v) \in (0,1)^2$ and $\Cb_C(s,s',u,v)=0$ else, 
where $\Bb_C$ denotes a tight, centered Gaussian process on $(\ell^\infty([0,1]^3), \|\cdot\|_\infty)$ with covariance kernel
$$
\Exp\{ \Bb_C(s,u,v)\Bb_C(s',u',v') \} = \min(s,s') [ C\{ \min(u,u'), \min(v,v') \} - C(u,v)C(u',v') ],
$$
and where $\dot C_j$, $j=1,2$, denotes the $j$th first-order partial derivative of $C$. 
\end{prop}

The limit process $\Db_C$ in~\eqref{eq:DbC} depends in an intractable way on the unknown copula~$C$, and, as a consequence, so do the limit distributions of the test statistics $S_{n,A}$ in~\eqref{eq:SnA} and $S_{n,A}(k^\star)$ in~\eqref{eq:SnAk}. The following \textit{multiplier bootstrap approximations}, initially proposed by \cite{Sca05} and \cite{RemSca09} in a copula setting, allow for the derivation of suitable approximations of the critical values. 

Let $B$ be some large integer and let $\xi_i^{(b)}$, $i=1,\dots,n$, $b=1,\dots,B$, denote i.i.d.\ standard normal random variables. Motivated by recent results in \cite{BucKoj14} and \cite{BucKojRohSeg14} on the multiplier bootstrap in a sequential setting, we can approximate the Gaussian process $\Bb_C$ appearing in the limit $\Db_C$ defined in~\eqref{eq:DbC} by 
\begin{multline}
\label{eq:checkBnb}
\check \Bb_n^{(b)}(s, s',u,v)
=
\frac{1}{\sqrt{n}}\sum_{i=\ip{ns}+1}^{\ip{ns'}}\xi_i^{(b)} \Big\{\ind \big( \hat{U}_{\ip{ns}+1:\ip{ns'},i} \le u, \hat{V}_{\ip{ns}+1:\ip{ns'},i} \le v \big) - C_{\ip{ns}+1:\ip{ns'}} (u,v)\Big\},
\end{multline}
for $(s,s',u,v) \in \Delta \times [0,1]^2$, where, for $1 \le k \le \ell \le n$, $C_{k:\ell}$ denotes the {\em empirical copula} based on the sample $(X_k, Y_k), \dots, (X_\ell, Y_\ell)$, that is,
\begin{align} \label{eq:empcop}
C_{k:\ell}(u,v) = \frac{1}{\ell-k+1} \sum_{i=k}^\ell \ind \big( \hat U_{k:\ell,i} \le u, \hat V_{k:\ell,i} \le v \big), \quad (u,v) \in [0,1]^2,
\end{align}
with the convention that $C_{k:\ell} = 0$ if $k > \ell$. If we replace $\Bb_C$ by 
$\check \Bb_n^{(b)}$ and all other unknown quantities in the definition of $\Db_C$ by natural estimators, then some standard calculations, carried out explicitly in the proof of 
the next proposition, suggest to define bootstrap replicates of $\Db_n$ as 
\begin{multline}
\label{eq:checkDnb}
\check  \Db_n^{(b)}(s,t)
=
\{ 1+ \hat A_{1:n}(t) \}^2  
\times \bigg\{
\frac{\ip{ns}}{n^{3/2}}\sum_{i=\ip{ns}+1}^n\xi_i^{(b)}\hat w_{\ip{ns}+1:n,i}(t)  -\frac{n-\ip{ns}}{n^{3/2}}\sum_{i=1}^{\ip{ns}}\xi_i^{(b)} \hat w_{1:\ip{ns},i}(t)
\bigg\},
\end{multline}
for $(s,t) \in [0,1]^2$, where,  for any $1 \le k \le \ell \le n$, 
\begin{multline}
\label{eq:wkl}
\hat w_{k:\ell,i}(t)
=
\widebar{m}_{k:\ell}(t)- \hat m_{k:\ell,i}(t) 
+
\{ \hat u_{k:\ell,i}(t)-\widebar{u}_{k:\ell} (t) \} \frac{ \hat a_{k:\ell}(t)}{ \hat b_{k:\ell}(t) }
+ \{ \hat v_{k:\ell,i}(t)-\widebar{v}_{k:\ell}(t) \} \frac{ \hat c_{k:\ell}(t) }{ \hat d_{k:\ell}(t)},
\end{multline}
with $\widebar m_{k:\ell}, \widebar u_{k:\ell}$ and $\widebar v_{k:\ell}$ denoting the arithmetic mean (over $i=k, \dots, \ell$) of
\begin{align*}
&\hat m_{k:\ell,i}(t)=\max(\hat{U}_{k:\ell,i}^{1/(1-t)},\hat{V}_{k:\ell,i}^{1/t}),  \qquad
&\hat u_{k:\ell,i}(t)=\hat{U}_{k:\ell,i}^{ \hat b_{k:\ell}(t) /(1-t)}, \qquad 
\hat v_{k:\ell,i}(t)=\hat{V}_{k:\ell,i}^{  \hat d_{k:\ell}(t) /t} ,
\end{align*}
respectively, and with
\begin{align*}
\hat a_{k:\ell}(t)&= \hat A_{k:\ell}(t)-t\hat {A}'_{k:\ell,n}(t), & \hat b_{k:\ell}(t)&= \hat A_{k:\ell}(t)+t, \\
\hat c_{k:\ell}(t)&= \hat A_{k:\ell}(t)+(1-t)\hat {A}'_{k:\ell,n}(t), & \hat d_{k:\ell}(t)&= \hat A_{k:\ell}(t)+1-t,
\end{align*}
where, for some positive sequence $h_n \downarrow 0$ such that $\inf_{n \geq 1} h_n \sqrt n >0$,
\begin{equation}
\label{eq:hatAkl'}
\hat {A}'_{k:\ell,n}(t) = \min[ \max \{ A'_{k:\ell,n}(t), -1 \}, 1], \quad t \in [0,1],
\end{equation}
with
\begin{equation}
\label{eq:Akl'}
A'_{k:\ell,n}(t)= \frac{1}{2h_n}\left \{ \hat A_{k:\ell}(t+h_n)- \hat A_{k:\ell}(t-h_n)\right \} ,
\end{equation}
for $t\in(h_n, 1-h_n)$, while $A'_{k:\ell,n}(t)=A'_{k:\ell,n} (h_n)$ for $t \le h_n$ and $A'_{k:\ell,n}(t)= A'_{k:\ell,n}(1-h_n)$ for $t \ge 1-h_n$.

The following proposition, proved in Appendix~\ref{sec:proofsboot}, establishes the asymptotic validity of the above resampling scheme under $\Hc_0$ in~\eqref{eq:H0}.

\begin{prop}\label{prop:bootdim2}
Under the conditions of Proposition~\ref{prop:weakdim2},
\[
\left( \Db_n, \check \Db_n^{(1)}, \dots, \check \Db_n^{(B)} \right)
\weak
\left(\Db_C,\Db_C^{(1)}, \dots, \Db_C^{(B)} \right)
\]
in $(\ell^\infty([0,1]^2), \| \cdot \|_\infty)^{B+1}$, where $\Db_C^{(1)}, \dots, \Db_C^{(B)}$ denote independent copies of~$\Db_C$.
\end{prop}

Bootstrap analogues of the test statistics in~\eqref{eq:SnA} and~\eqref{eq:SnAk} can be defined by replacing $\Db_n$ by $\check \Db_n^{\scriptscriptstyle (b)}$ in the corresponding definitions. Through the continuous mapping theorem, we immediately obtain that, under $\Hc_0$, the random vector
$ 
( S_{n,A}, \check S_{n,A}^{\scriptscriptstyle (1)}, \dots, \check S_{n,A}^{\scriptscriptstyle (B)}  ) \in \mathbb{R}^{B+1}
$
weakly converges to a vector with i.i.d.\ components, each component having the same distribution as $\sup_{s \in [0,1]} \int_{[0,1]}  \Db_C^2(s,t)\, d\mu(t)$. Hence, a test rejecting $\Hc_0$ at the significance level $\alpha$ if $S_{n,A}$ exceeds the $\ip{(1-\alpha)B}$-th order statistic of $\check S_{n,A}^{\scriptscriptstyle (1)} , \dots, \check S_{n,A}^{\scriptscriptstyle (B)}$ asymptotically keeps its level for $n \to \infty$ followed by $B \to \infty$ (see Appendix~F in \citealp{BucKoj14}). 

At the cost of a more complex notation, the results presented in this section can be extended to the case $d > 2$. The main steps are given in Appendix~\ref{sec:extensiond>2}.

\section{Test statistics for $d=2$ under known marginal breaks} \label{sec:testdim2break}
\def\theequation{4.\arabic{equation}}
\setcounter{equation}{0}

It is only if  $\Hc_{0,m}$ in~\eqref{eq:H0m} holds that the tests developed in the previous section can be used to reject $\Hc_{0,c}$ in~\eqref{eq:H0A}. Empirical evidence of the latter fact will be given in Section~\ref{sec:sim}. In some applications, such as those presented in Section~\ref{sec:illustration}, there might be reasons to believe that~\eqref{eq:H0m} does not hold. The aim of this section is to propose an adaptation of the tests derived in the previous section that will be consistent against $\Hc_{0,c}$ in~\eqref{eq:H0A} when there are known abrupt changes in the margins. 

Recall that we assume to observe an independent sequence $(X_i, Y_i)$, $i=1,\dots,n$, where $(X_i, Y_i)$ has copula $C^{(i)}$ and continuous marginal c.d.f.s $F^{(i)}$ and $G^{(i)}$, respectively. Before considering a more general framework, we shall first focus, for pedagogical reasons, on the following simple hypothesis for the margins:
\begin{equation}
\label{eq:H1m}
\Hc_{1,m}: \left\{
\begin{array}{l}
\mbox{there exists a known } \theta \in (0,1) \mbox{ such that} \\
F^{(1)} = \dots = F^{(\ip{n\theta})} \neq F^{(\ip{n\theta}+1)} = \dots = F^{(n)},  \\ 
G^{(1)} = \dots = G^{(\ip{n\theta})} \neq G^{(\ip{n\theta}+1)} = \dots = G^{(n)},
\end{array}
\right.
\end{equation}
that is, both margins change abruptly at the same time point $\ip{n\theta}$. As we continue, we shall also use the notation $m = \ip{n\theta}$. In the illustrative applications presented in Section~\ref{sec:illustration}, this time point corresponds to the construction of a dam on a river. 

In order to adapt the testing procedures derived in the previous section to the above setting, we need to propose an appropriate version of the test statistics $S_{n,A}$ in~\eqref{eq:SnA} and $S_{n,A}(k^\star)$ in~\eqref{eq:SnAk} under $\Hc_{1,m} \cap \Hc_{0,c}$, where $\Hc_{0,c}$ is defined in~\eqref{eq:H0A}. Mimicking the developments carried out in the previous section, it is natural to start from adequate versions of the estimators in~\eqref{eq:hatSkl}. Hence, assume that $\Hc_{1,m} \cap \Hc_{0,c}$ holds, and let $(X_i, Y_i)$, $i=k, \dots, \ell$, be some subsample. If $k < m < \ell$, a natural estimator for $S$ in~\eqref{eq:St} is given by a convex combination of the estimators defined in~\eqref{eq:hatSkl}, one based on the subsample from $k$ to $m$, the other on the subsample from $m+1$ to~$\ell$. In other words, meaningful analogues of the estimators in~\eqref{eq:hatSkl} are, for any $t \in [0,1]$,
\begin{align} \label{eq:Sklbreak} 
\hat S_{k:\ell}^\theta(t) 
= 
\begin{cases}
\frac{m-k+1}{\ell-k+1} \hat S_{k:m}(t) + \frac{\ell-m}{\ell-k+1} \hat S_{m+1:\ell}(t), &\mbox{if } k < m < \ell, \\
\hat S_{k:\ell}(t), & \mbox{otherwise}.
\end{cases}
\end{align}
Proceeding as in the previous section, the corresponding subsample estimators of $A$ are then simply
\begin{equation}\label{eq:Aklbreak} 
\hat A_{k:\ell}^\theta(t) 
=
\frac{
\hat S_{k:\ell}^\theta(t) 
}{
1 - \hat S_{k:\ell}^\theta(t) 
}
, \quad t \in [0,1].
\end{equation}
It is easy to verify that the formulas in~\eqref{eq:Sklbreak} and~\eqref{eq:Aklbreak} coincide with those in~\eqref{eq:hatSkl} and~\eqref{eq:hatAkl}, respectively, provided that the pseudo-observations in~\eqref{eq:pseudoobs} are replaced by appropriates ones taking into account the break in the margins at time point $m$, namely, for $i = k, \dots, \ell$, \begin{equation}
\label{eq:pseudotheta}
\hat U_{\theta,k:\ell,i} = 
\begin{cases}
\hat U_{k:\ell, i} & \text{ if } m \notin \{k, \dots \ell\} ,\\
\hat U_{k:m,i} & \text{ if } m \in \{ k, \dots, \ell\} \text{ and } i \le m, \\
\hat U_{m+1:\ell,i} & \text{ if } m \in \{ k, \dots, \ell\} \text{ and } i > m, 
\end{cases}
\end{equation}
and similarly for $\hat V_{\theta,k:\ell,i}$. Thus, from a practical perspective, once the above adapted pseudo-observations are computed, the computer code for the simpler setting considered in the previous section can be fully reused.

It follows that natural generalizations of the process $\Db_n$ and the statistics $S_{n,A}$ and $S_{n,A}(k^\star)$, defined in~\eqref{eq:Dn},~\eqref{eq:SnA} and~\eqref{eq:SnAk}, respectively, are simply
\[
\nonumber
\Db_n^\theta(s,t) = \frac{\ip{ns} (n - \ip{ns})}{n^{3/2}} \left\{\hat A_{1:\ip{ns}}^\theta(t) - \hat A_{\ip{ns}+1:n}^\theta(t) \right\}, \quad (s,t) \in [0,1]^2
\]
and
\begin{equation}
\label{eq:SnAtheta}
S_{n,A}^\theta = \max_{1\leq k<n}\int_{[0,1]}\left\{ \Db^\theta_n(k/n,t)\right \}^2\, d\mu(t), \qquad
S_{n,A}^\theta(k^\star) = \int_{[0,1]}\left\{ \Db^\theta_n(k^\star/n,t)\right \}^2\, d\mu(t),
\end{equation}
respectively. 

The following proposition, proved in Appendix~\ref{sec:proofs}, establishes the limit distribution of $\Db_n^\theta$ under $\Hc_{0,c} \cap \Hc_{1,m}$, where $\Hc_{0,c}$ and $\Hc_{1,m}$ are defined in~\eqref{eq:H0A} and~\eqref{eq:H1m}, respectively.

\begin{prop}\label{prop:weakdim2break}
Assume that $\Hc_{0,c} \cap \Hc_{1,m}$ holds and that the Pickands dependence function $A$ associated with $C$ is continuously differentiable on $(0,1)$. Then, in the normed space $(\ell^\infty([0,1]^2), \| \cdot \|_\infty)$, $\Db_n^\theta \weak \Db_C$, where $\Db_C$ is defined in~\eqref{eq:DbC}.
\end{prop}

Hence, we see that the limit distribution of the process $\Db_n^\theta$ under $\Hc_{0,c} \cap \Hc_{1,m}$ does not depend on the marginal break point $\theta$ and coincides with that of the process $\Db_n$ in~\eqref{eq:Dn} whose asymptotics were studied in Proposition~\ref{prop:weakdim2} under $\Hc_{0,c} \cap \Hc_{0,m}$, where $\Hc_{0,m}$ is defined in~\eqref{eq:H0m}.

To carry out the tests based on $S_{n,A}^\theta$ and $S_{n,A}^\theta(k^\star)$, we need bootstrap replicates of the pro\-cess~$\Db_n^{\scriptscriptstyle \theta}$. The latter are denoted by $\check \Db_n^{\scriptscriptstyle \theta,(b)}$, $b=1,\dots,B$, and are defined exactly as in~\eqref{eq:checkDnb}, except for the underlying pseudo-observations which are computed as in~\eqref{eq:pseudotheta}. The next result, proved in Appendix~\ref{sec:proofsboot}, establishes the asymptotic validity of the proposed resampling scheme under $\Hc_{0,c} \cap \Hc_{1,m}$.

\begin{prop}\label{prop:bootdim2break}
Under the conditions of Proposition~\ref{prop:weakdim2break}, in $(\ell^\infty([0,1]^2), \| \cdot \|_\infty)^{B+1}$,
\[
\left( \Db_n^\theta, \check \Db_n^{\theta,(1)}, \dots, \check \Db_n^{\theta,(B)} \right)
\weak
\left(\Db_C,\Db_C^{(1)}, \dots, \Db_C^{(B)} \right),
\]
where $\Db_C$ is defined in~\eqref{eq:DbC} and $\Db_C^{(1)}, \dots, \Db_C^{(B)}$ denote independent copies of $\Db_C$.
\end{prop}

Multiplier bootstrap replicates of the test statistics in~\eqref{eq:SnAtheta} 
are defined {\em mutatis mutandis}, the only difference with the formulas of Section~\ref{sec:testdim2} being again the pseudo-observations which are computed using~\eqref{eq:pseudotheta}. As previously, the null hypothesis is rejected at the significance level $\alpha$ if $S_{n,A}^\theta$ (resp.\ $S_{n,A}^\theta(k^\star)$) exceeds the $\ip{(1-\alpha)B}$-th order statistic computed from the sample of its $B$ multiplier bootstrap replicates.

We end this section by three remarks on further possible extensions.

\begin{remark} [More than one marginal break]
In the preceding parts of this section, we restricted ourselves to the case of exactly one abrupt break that affects both marginal distributions. At the cost of a more complex notation, the previous results can all be simply extended to a more general setting. For $R\in\N$, $R \geq 1$, let  
\[
\Theta_R 
= 
\{(\theta_0, \theta_1, \dots, \theta_{R+1}) \in \R^{R+2}:  0 = \theta_0 < \theta_1 < \dots < \theta_R < \theta_{R+1}=1\}.
\]
A vector $\vect \theta \in \Theta_R$ should be interpreted as encoding (rescaled) time points where at least one of the marginal distributions changes abruptly.  Specifically, let
\begin{equation*}
\Hc_{1,m}^{(R)}: \left\{
\begin{array}{l}
\mbox{there exists a known } \vect \theta \in \Theta_R \mbox{ such that, for any } r=0, \dots, R, \\
F^{(\ip{n\theta_{r}}+1)}=\ldots=F^{(\ip{n\theta_{r+1}})} \\ 
G^{(\ip{n\theta_{r}}+1)}=\ldots=G^{(\ip{n\theta_{r+1}})} \\
\text{and such that, for at least one of the margins, } \\ 
\text{the c.d.f.s at $\ip{n\theta_r}$ and $\ip{n\theta_r}+1$ are different.}
\end{array}
\right.
\end{equation*}
Identifying the parameter $\theta\in(0,1)$ used previously in this section with the vector $\vect \theta=(0,\theta,1)\in \Theta_1$, we see that $\Hc_{1,m}$ in \eqref{eq:H1m} implies $\Hc_{1,m}^{(1)}$ above. To derive an extension of the test based on $S_{n,A}^\theta$ in~\eqref{eq:SnAtheta} that can be used to reject $\Hc_{0,c}$ in~\eqref{eq:H0A} under $\Hc_{1,m}^{\scriptscriptstyle (R)}$ above, it merely suffices to adapt the definition of the pseudo-observations in~\eqref{eq:pseudotheta} to this more complex situation. For that purpose, let $m_r=\ip{n\theta_r}$ for all $r=0,\dots,R+1$, and, for any $i=1, \dots, n$, let 
\begin{align*}
m_\wedge=m_\wedge(i)
&=
\max\{ m_r \in \{m_{0}, \dots, m_R\}: m_{r} < i \}, \\
m_\vee=m_\vee(i) 
&=
\min\{ m_r \in \{m_{1}, \dots, m_{R+1}\}: m_{r} \ge i \},
\end{align*} 
such that $i \in \{ m_\wedge+1, \dots, m_\vee \}$. Then, analogously to~\eqref{eq:pseudotheta}, we can define pseudo-observations adapted to $\Hc_{1,m}^{(R)}$, for $1\le k \le \ell \le n$ and $i=k,\dots, \ell$, as
\begin{equation*}
\hat U_{\vect \theta,k:\ell,i} = 
\begin{cases}
\hat U_{k:\ell, i} & \text{ if } m_\wedge < k \text{ and } \ell < m_\vee, \\
\hat U_{m_\wedge+1:\ell, i} & \text{ if } m_\wedge \ge k \text{ and } \ell < m_\vee, \\
\hat U_{k:m_\vee, i} & \text{ if } m_\wedge < k \text{ and } \ell \ge m_\vee, \\
\hat U_{m_\wedge+1:m_\vee, i} & \text{ if } m_\wedge \ge k \text{ and } \ell \ge m_\vee, \\
\end{cases}
\end{equation*}
and similarly for $\hat V_{\vect \theta,k:\ell,i}$. All the formulas from Section~\ref{sec:testdim2} remain identical up to the use of the above definition for the pseudo-observations. In addition, the proofs of the theoretical results stated in this section extend easily but are more cumbersome to write.  For the sake of brevity, we omit further details.
\end{remark}

\begin{remark}[Extension to PQD copulas]
The previously studied tests could be extended to copulas that are not necessary of the extreme-value type, for instance to {\em positive quadrant dependent} (PQD) copulas \citep[see, e.g.,][Chapter~5]{Nel06}, that is, copulas satisfying $C(u,v) \geq uv$ for all $(u,v) \in [0,1]^2$. Starting from~\eqref{eq:evc}, it can be easily verified that extreme-value copulas are PQD. The extension of the tests to PQD copulas relies on the fact (used in the proofs given in Appendices~\ref{sec:proofs} and~\ref{sec:proofsboot}) that the Pickands dependence function associated with an extreme-value copula $C$ can be expressed as
$$
A(t) = \frac{1 - \int_0^1 C(y^{1-t},y^t) \, dy}{\int_0^1 C(y^{1-t},y^t) \, dy}, \quad t \in [0,1].
$$
The previous functional of $C$ remains well-defined outside of the extreme-value framework as long as $\int_0^1 C(y^{1-t},y^t) dy \neq 0$. The PQD condition, for instance, ensures that $C(y^{1-t},y^t) \ge y$ for all $(y,t) \in (0,1) \times [0,1]$, and therefore the existence of~$A$.

The only change necessary to extend the previously studied tests to PQD copulas concerns the estimation of the partial derivatives $\dot C_1$ and $\dot C_2$ of the copula $C$ that is required for carrying out the multiplier resampling scheme. The approach adopted in Section~\ref{sec:testdim2} (see~\eqref{eq:Akl'}) could be replaced by the more classical one consisting of using finite-differences based on the empirical copula. This would however lead to significantly less convenient formulas as far as implementation is concerned. Given that the meaning of the functional $A$ is unclear outside of the extreme-value framework and that tests for change-point detection for non extreme-value copulas already exist \citep[see, e.g.,][]{BucKojRohSeg14,KojQueRoh15}, we do not pursue this further.
\end{remark}

\begin{remark}[Estimation of the marginal change-points] An even more general framework regarding marginal change-points is to assume that their number is known but not their position. For instance, under the assumption of one unknown marginal (scaled) change-point $\theta \in (0,1)$ in the first margin, its value could be estimated under some additional assumption on the nature of the change-point (change in mean, in variance, etc). A sensible estimator $\hat \theta_n$ of $\theta$ is then given by $\hat k/n$, where $\hat k$ is the value of $k \in \{1,\dots,n-1\}$ that maximizes a corresponding natural max-type change-point statistic. Under the additional assumption that  $\hat \theta_n - \theta = O_\Prob(1/n)$ \citep[see, e.g.,][]{Dum91}, we conjecture that the limiting behavior of the statistic $S_{n,A}^{\scriptscriptstyle \hat \theta_n}$ is the same as that of $S_{n,A}^\theta$ in~\eqref{eq:SnAtheta}. The theoretical justification seems to be quite involved, and, since we do not need this extension for the hydrological applications presented in Section~\ref{sec:illustration}, we do not pursue this further.
\end{remark}


\section{Simulation study}\label{sec:sim}

Simulations were carried out in order to evaluate the finite-sample performance of the tests studied in Sections~\ref{sec:testdim2} and~\ref{sec:testdim2break}. For the sake of simplicity, we only focused on the tests based on $S_{n,A}$ in~\eqref{eq:SnA} and $S_{n,A}^{\scriptscriptstyle \theta}$ in~\eqref{eq:SnAtheta}, as the finite-sample behavior of the corresponding two-sample tests should be strongly related. Recall that the tests based on $S_{n,A}$ and $S_{n,A}^{\scriptscriptstyle \theta}$ are procedures for testing $\Hc_0$ in~\eqref{eq:H0} designed to be particularly sensitive to departures from $\Hc_{0,c}$ in~\eqref{eq:H0A}. The former (resp.\ latter) should not be used to reject $\Hc_{0,c}$ unless $\Hc_{0,m}$ in~\eqref{eq:H0m} (resp.\ $\Hc_{1,m}$ in~\eqref{eq:H1m}) holds.

The rejection rates of the two tests were estimated from samples drawn from bivariate distributions whose copulas are of the form
\begin{align}
\label{eq:khoudraji}
C_{\mathbf{a},\vartheta}(u,v)=u^{a_1}v^{a_2}C_\vartheta(u^{1-a_1},v^{1-a_2}), \quad (u,v) \in [0,1]^2,
\end{align}
where $C_\vartheta$ is a symmetric extreme-value copula with parameter $\vartheta \in \R$ and $\mathbf{a}=(a_1,a_2)\in[0,1]^2$ is a parameter controlling the amount of asymmetry. The above copula construction principle is frequently referred to as {\em Khoudraji's device} \citep{Kho95,GenGhoRiv98,Lie08}. As long as $C_\vartheta$ is an extreme-value copula, so is its potentially asymmetric version $C_{\mathbf{a},\vartheta}$. Given that there is hardly any practical difference among the existing bivariate symmetric parametric families of extreme-value copulas \citep[see][for more details]{GenKojNesYan11}, $C_\vartheta$ in~\eqref{eq:khoudraji} was taken to be the Gumbel--Hougaard copula with parameter $\vartheta$.

The finite-sample performance of the tests based on $S_{n,A}$ and $S_{n,A}^\theta$ was compared with that of two other tests for $\Hc_0$ designed to be particularly sensitive to $\Hc_{0,c}$ in~\eqref{eq:H0c}: the test based on the empirical copula studied in \citet{BucKojRohSeg14} (statistic $\check{S}_n$) and the test based on Spearman's rho considered in \citet{KojQueRoh15} (statistic $\tilde{S}_{n,1}$). Both tests rely on multiplier bootstraps for the computation of approximate p-values. They will be referred to as {\em the test based on $S_{n,C}$} and {\em the test based on $S_{n,\rho}$}, respectively. These procedures however do not assume the underlying dependence structures to be of the extreme-value type. The former is sensitive to all kind of changes in the underlying copula, while the latter is only sensitive to changes in Spearman's rho. As for the test based on $S_{n,A}$, they should not be used to reject $\Hc_{0,c}$ unless $\Hc_{0,m}$ in~\eqref{eq:H0m} holds. 

All the tests considered in our numerical experiments were carried out at the 5\% significance level using $B=1000$ multiplier bootstrap replicates. The values 50, 100 and 200 were considered for the sample size $n$. The measure $\mu$ involved in the definition of $S_{n,A}$ in~\eqref{eq:SnA} and $S_{n,A}^\theta$ in~\eqref{eq:SnAtheta} was taken equal to $9^{-1}\sum_{i=1}^9\delta_{i/10}$. The bandwidth $h_n$ in~\eqref{eq:Akl'} was set to $10^{-2}/\sqrt{n}$. With the illustrations of Section~\ref{sec:illustration} in mind, the values 0.25 and 0.5 were considered for $\theta$.  The computations were carried out using the \textsf{R} statistical system \citep{Rsystem}, and the \textsf{R} packages {\tt copula} \citep{copula} and {\tt npcp}~\citep{npcp}.

\begin{table}
\begin{tabular}{cccc|ccccc}
\toprule[1.5pt]
$n$ & $\mathbf{a}$ & $\vartheta$ & $\tau$ & $S_{n,A}$  & $S_{n,C}$ & $S_{n,\rho}$ & $S_{n,A}^{0.25}$ & $S_{n,A}^{0.5}$ \\ 
\midrule[1pt]
50 & $(0,0)$ & 1 & 0 & 4.9  & 6.3 & 5.6 & 7.6 & 4.2 \\ 
 &  & 1.25 & .2 & 6.7 & 6.2 & 5.8 & 7.9 & 7.0  \\ 
 &  & 1.67 & .4 & 5.8  & 4.4 & 3.9 & 6.6 & 6.1 \\ 
 &  & 2.5 & .6 & 4.0  & 3.7 & 2.4 & 5.6 & 4.7 \\ 
 &  & 5 & .8 & 3.6  & 2.4 & 0.9 & 8.2 & 2.7 \\ 
 & $(0,.3)$ & 4 & .56 & 4.5  & 4.7 & 3.1 & 5.2 & 5.5\\ 
 \midrule[1pt]
100 & $(0,0)$ & 1 & 0 & 5.5  & 5.1 & 5.8 & 7.7 & 5.4\\ 
 &  & 1.25 & .2 & 6.3 & 5.4 & 6.2 & 7.4 & 6.9 \\ 
 &  & 1.67 & .4 & 6.2 & 4.3 & 4.4 & 6.2 & 6.6  \\ 
 &  & 2.5 & .6 & 5.4 & 3.0 & 2.9 & 6.0 & 5.5 \\ 
 &  & 5 & .8 & 2.0 & 2.2 & 1.0 & 4.0 & 2.6 \\ 
 & $(0,.3)$ & 4 & .56 & 4.5 & 4.2 & 3.8 & 4.5 & 5.0 \\ 
 \midrule[1pt]
200 & $(0,0)$ & 1 & 0 & 5.0 & 4.3 & 4.8 & 6.2 & 5.6 \\ 
 &  & 1.25 & .2 & 6.0 & 4.8 & 5.8 & 6.4 & 6.4  \\ 
 &  & 1.67 & .4 & 5.9 & 4.0 & 4.9 & 6.4 & 6.2  \\ 
 &  & 2.5 & .6 & 3.6 & 2.8 & 3.1 & 4.4 & 4.4  \\ 
 &  & 5 & .8 & 2.6  & 1.3 & 2.0 & 3.4 & 3.4 \\ 
 & $(0,.3)$ & 4 & .56 & 4.8  & 4.0 & 4.3 & 5.2 & 5.2 \\
  \bottomrule[1.5pt]
\end{tabular}
\centering
\caption{Rejection rates of $\Hc_0$ in \% estimated from 4000 random samples generated under $\Hc_0$ from c.d.f.\ $C_{\mathbf{a},\vartheta}$ in~\eqref{eq:khoudraji}. The column $\tau$ gives the value of Kendall's tau of the copula $C_{\mathbf{a},\vartheta}$.}
\label{table:nullsize} 
\end{table}

\vspace{-.3cm}
\paragraph{Empirical levels of the tests based on $S_{n,A}$, $S_{n,C}$ and $S_{n,\rho}$} Columns 5-7 of Table~\ref{table:nullsize} report the rejection rates of the three tests estimated from 4000 random samples generated under $\Hc_0$ from c.d.f.\ $C_{\mathbf{a},\vartheta}$ in~\eqref{eq:khoudraji} for various values of $\mathbf{a}$ and $\vartheta$. As one can see, the empirical levels are overall reasonably close to the 5\% nominal level in all settings for which Kendall's tau $\tau$ of $C_{\mathbf{a},\vartheta}$ is strictly smaller than 0.6. For $\tau \geq 0.6$, the three tests are overall too conservative, although, as expected, the empirical levels improve as $n$ increases. 

\begin{figure}
\centering
\includegraphics[width=.8\linewidth]{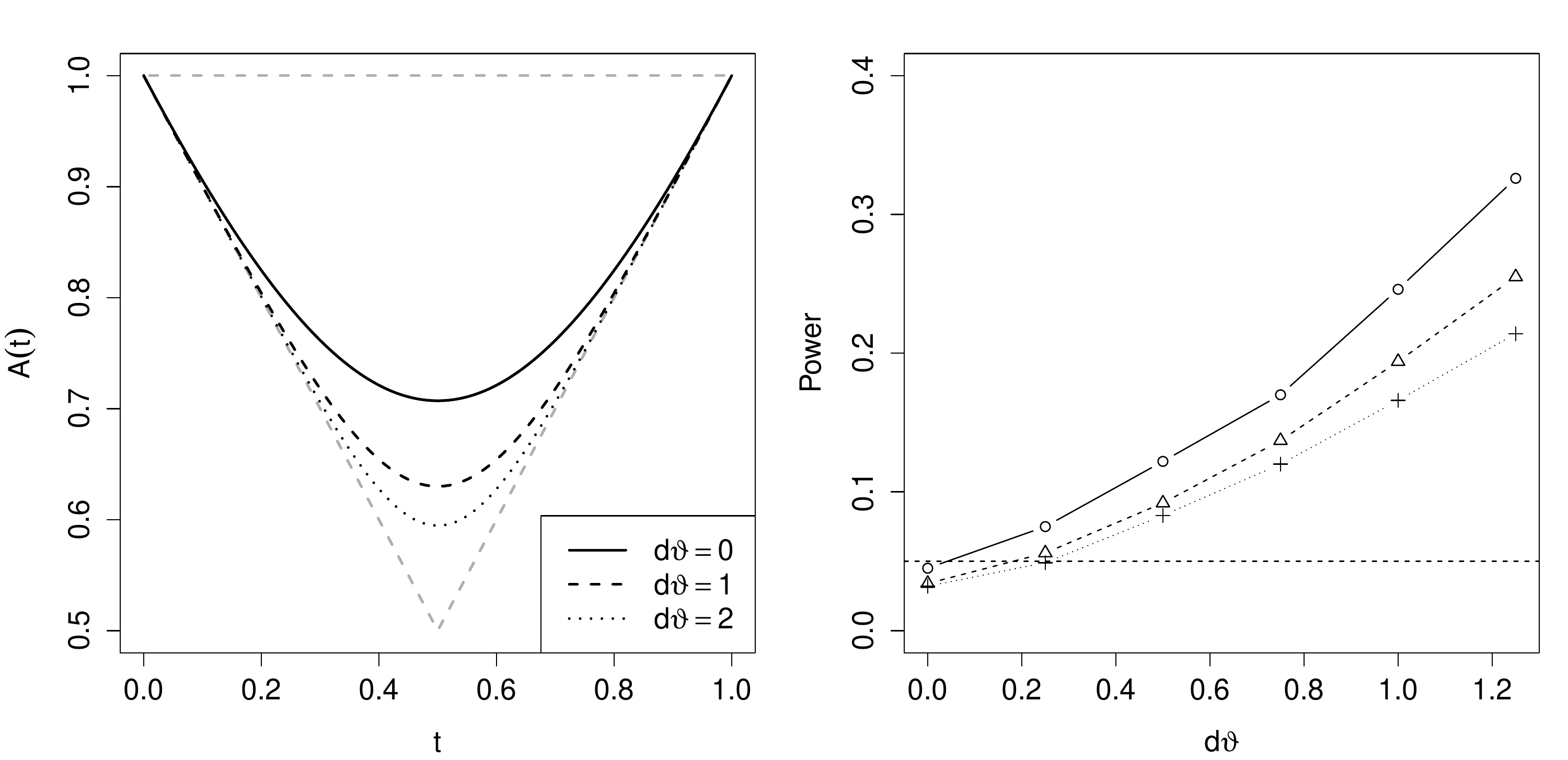}
\caption{Left: Pickands dependence function of the Gumbel--Hougaard copula with parameter $2+d\vartheta$. Right: Rejection rates of the tests based on $S_{n,A}$ $(\circ)$, $S_{n,C}$ $(\triangle)$ and $S_{n,\rho}$ $(+)$ versus $d\vartheta$ estimated from 2000 bivariate samples of size $n=100$ such that, for each sample, the first (resp.\ last) 50 observations were drawn from a Gumbel--Hougaard copula with parameter 2 (resp.\ $2+d\vartheta$).}
\label{fig:powerdgp1}
\end{figure}

\begin{figure}
\centering
\includegraphics[width=.8\linewidth]{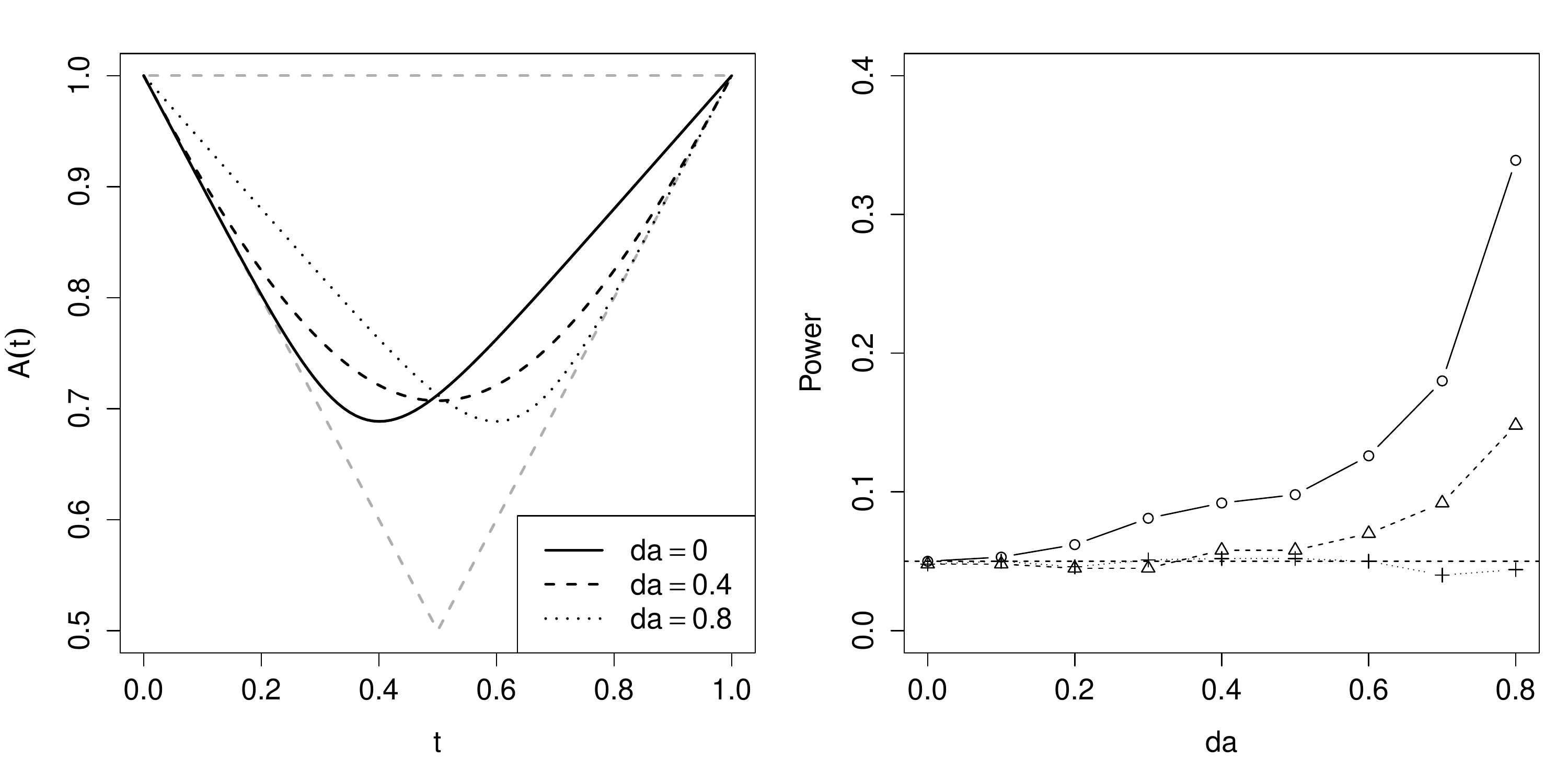}
\caption{Left: Pickands dependence function associated with the copula $C_{\mathbf{a},\vartheta}$ in~\eqref{eq:khoudraji} with $\mathbf{a}=(\max(0.4-da,0),\max(da-0.4,0))$, for $da \in \{0,0.4,0.8\}$, and $\vartheta$ set to keep Kendall's tau of $C_{\mathbf{a},\vartheta}$ equal to $0.5$. Right: Rejection rates of the tests based on $S_{n,A}$ $(\circ)$, $S_{n,C}$ $(\triangle)$ and $S_{n,\rho}$ $(+)$ versus $da$ estimated from 2000 bivariate samples of size $n=200$ such that, for each sample, the first (resp.\ last) 100 observations were drawn the above mentioned copula with $da=0$ (resp.\ $da \in \{0,0.1,\dots,0.8\}$).}
\label{fig:powerasymmetricgumbel}
\end{figure}

\vspace{-.3cm}
\paragraph{Empirical power of the tests based on $S_{n,A}$, $S_{n,C}$ and $S_{n,\rho}$ under changes in the copula only} The right plot of Figure~\ref{fig:powerdgp1} displays the rejection rates of the three tests estimated from 2000 samples of size $n=100$ generated under $\Hc_{0,m}\cap \neg \Hc_{0,c}$ such that, for each sample, the first (resp.\ last) 50 observations were drawn from c.d.f.~\eqref{eq:khoudraji} with $(\mathbf{a},\vartheta)=(0,0,2)$ (resp.\ $(0,0,2+d\vartheta)$). As expected, the test based on $S_{n,A}$ is more powerful than its two competitors in this simple setting. 

To investigate the influence of asymmetry on the power of the three tests, as a second experiment, we considered again the copula $C_{\mathbf{a},\vartheta}$ in~\eqref{eq:khoudraji}, but this time with parameter $\mathbf{a}$ defined as $(\max(0.4-da,0),\max(da-0.4,0))$, for $da \in \{0,0.1,\dots,0.8\}$, and with parameter $\vartheta$ set to keep Kendall's tau of $C_{\mathbf{a},\vartheta}$ equal to $0.5$. The corresponding Pickands dependence functions for $da \in \{0,0.4,0.8\}$ are represented in the left plot of Figure~\ref{fig:powerasymmetricgumbel}. The right plot displays the rejection rates of the tests based on $S_{n,A}$, $S_{n,C}$ and $S_{n,\rho}$ versus $da$ estimated from 2000 samples of size $n=200$ such that, for each sample, the first (resp.\ last) 100 observations were drawn from the above mentioned copula with $da=0$ (resp.\ $da \in \{0,0.1,\dots,0.8\}$). Although the rejection rates are overall relatively low, the test based $S_{n,A}$ is by far the best. The fact that the test based on $S_{n,\rho}$ has no power against such alternatives is due to the fact that Spearman's remains almost constant. 

\begin{table}
\centering
\begin{tabular}{ccc|ccc|ccc}
\toprule[1.5pt]
\multicolumn{3}{c}{} & \multicolumn{3}{|c|}{$\theta=0.5$} & \multicolumn{3}{c}{$\theta=0.25$}   \\
$d\mu$ & $n$ & $\tau$ & $S_{n,A}$ & $S_{n,C}$ & $S_{n,\rho}$ & $S_{n,A}$ & $S_{n,C}$ & $S_{n,\rho}$  \\
\midrule[1pt]
5 & 50 & 0 & 5.5 & 6.2 & 6.0 & 5.5 & 6.3 & 6.2 \\
& & 0.25 & 6.6 & 6.1 & 6.2 & 5.2 & 4.1 & 3.7 \\
& & 0.5 & 4.3 & 3.3 & 2.5 & 4.4 & 2.8 & 2.7 \\
& & 0.75 & 3.9 & 2.2 & 1.2 & 2.7 & 2.5 & 1.1  \\
& 100 & 0 & 5.1 & 5.1 & 5.4 & 5.0 & 5.9 & 6.0 \\
& & 0.25 & 5.1 & 3.6 & 5.1 & 6.7 & 4.0 & 4.9 \\
& & 0.5 & 5.5 & 4.1 & 4.3 & 4.7 & 2.9 & 3.3  \\
& & 0.75 & 3.3 & 1.0 & 0.6 & 4.5 & 2.8 & 1.5 \\
& 200 & 0 & 5.6 & 4.8 & 5.5 & 5.6 & 5.8 & 5.4 \\
& & 0.25 & 4.8 & 5.1 & 5.1 & 3.9 & 3.6 & 3.4 \\
& & 0.5 & 4.6 & 2.5 & 2.9 & 5.2 & 3.4 & 4.0  \\
& & 0.75 & 4.0 & 0.8 & 1.0 & 5.0 & 3.0 & 2.1  \\
\midrule[1pt]
15 & 50 & 0 & 4.5 & 5.6 & 5.3 & 4.7 & 5.3 & 5.3  \\
& & 0.25 & 7.6 & 5.7 & 5.0 & 6.8 & 6.5 & 6.8 \\
& & 0.5 & 5.2 & 2.8 & 2.1 & 8.6 & 6.7 & 5.3 \\
& & 0.75 & 18.3 & 1.5 & 0.3 & 20.0 & 15.9 & 4.6  \\
& 100 & 0 & 4.3 & 4.4 & 4.9 & 4.2 & 4.2 & 4.6 \\
& & 0.25 & 4.7 & 3.4 & 3.9 & 5.9 & 4.3 & 4.9 \\
& & 0.5 & 7.2 & 3.8 & 3.0 & 9.0 & 8.4 & 6.5  \\
& & 0.75 & 40.6 & 10.0 & 1.8 & 42.0 & 36.9 & 23.0  \\
& 200 & 0 & 4.3 & 3.8 & 5.4 & 4.2 & 4.3 & 4.3  \\
& & 0.25 & 6.4 & 5.5 & 5.2 & 7.7 & 5.5 & 6.5 \\
& & 0.5 & 9.3 & 7.4 & 6.3 & 14.1 & 15.7 & 14.7  \\
& & 0.75 & 75.2 & 56.5 & 36.8 & 79.2 & 79.3 & 71.5  \\
\bottomrule[1.5pt]
\end{tabular}
\caption{Rejection rates of $\Hc_0$ in \% estimated from 1000 bivariate samples of size $n$ generated under $\Hc_{1,m} \cap \Hc_{0,c}$, where $\Hc_{0,c}$ and $\Hc_{1,m}$ are defined in~\eqref{eq:H0c} and~\eqref{eq:H1m}, respectively, such that, for each sample, the first $\ip{n\theta}$ (resp.\ last $n -\ip{n\theta}$) observations were drawn from a c.d.f.\ whose copula is the Gumbel--Hougaard, whose first margin is GEV with parameters $\mu=20$, $\sigma=10$ and $\gamma=0.25)$ (resp.\ $\mu=20+d\mu$, $\sigma=10$ and $\gamma=0.25$), and whose second margin is standard normal. The value of the parameter of the Gumbel--Hougaard copula is set through its one-to-one relationship with Kendall's tau $\tau$.}
\label{table:H1m} 
\end{table}

\vspace{-.3cm}
\paragraph{Empirical power of the tests based on $S_{n,A}$, $S_{n,C}$ and $S_{n,\rho}$ under an abrupt change in one margin only} Table~\ref{table:H1m} reports rejection rates of $\Hc_0$ estimated from 1000 bivariate samples of size $n$ generated under $\Hc_{1,m} \cap \Hc_{0,c}$ where $\Hc_{0,c}$ and $\Hc_{1,m}$ are defined in~\eqref{eq:H0c} and~\eqref{eq:H1m}, respectively, such that, for each sample, the first $\ip{n\theta}$ (resp.\ last $n -\ip{n\theta}$) observations were drawn from a c.d.f.\ whose copula is the Gumbel--Hougaard, whose first margin is GEV with parameters $\mu=20$, $\sigma=10$ and $\gamma=0.25$ (resp.\ $\mu=20+d\mu$, $\sigma=10$ and $\gamma=0.25$), and whose second margin is standard normal (the results are unaffected by the choice of the second margin since the test is rank-based). 

All three tests have little power against such alternatives when the shift $d\mu$ in the location parameter of the first margin is relatively small ($d\mu=5$). This is a desirable property since the tests were designed to be sensitive to departures from $\Hc_{0,c}$. Higher rejection rates were obtained for $d\mu=15$ and when the dependence is moderate or high, in particular if the (scaled) change-point in the first margin is non-central $(\theta=0.25)$. The latter results illustrate the fact that the procedures based on $S_{n,A}$, $S_{n,C}$ and $S_{n,\rho}$ are tests for $\Hc_0$ and that one should not use them to reject $\Hc_{0,c}$ unless $\Hc_{0,m}$ holds. Additional changes in the dispersion or scale parameter of the first margin might even increase the phenomenon.

\vspace{-.3cm}
\paragraph{Empirical levels of the test based on $S_{n,A}^\theta$} A consequence of Proposition~\ref{prop:bootdim2break} is that the test based on $S_{n,A}^\theta$ will hold its level asymptotically under one abrupt marginal change only, such as those considered in the previous experiment. To evaluate the corresponding finite-sample behavior, we considered again the setting of Table~\ref{table:nullsize}. Indeed, because of the rank-based nature of the test based on $S_{n,A}^\theta$, samples generated under $\Hc_0$ can equivalently be regarded as generated from $\Hc_{0,c} \cap \Hc_{1,m}$. From the last two columns of Table~\ref{table:nullsize},  we see that the test based on $S_{n,A}^{0.5}$ holds its level equally well as the test based on $S_{n,A}$. The test based on $S_{n,A}^{0.25}$ is however slightly too liberal for $n=50$, although the agreement of its empirical levels with the 5\% nominal level improves as $n$ increases.

\begin{figure}
\centering
\includegraphics[width=1\linewidth]{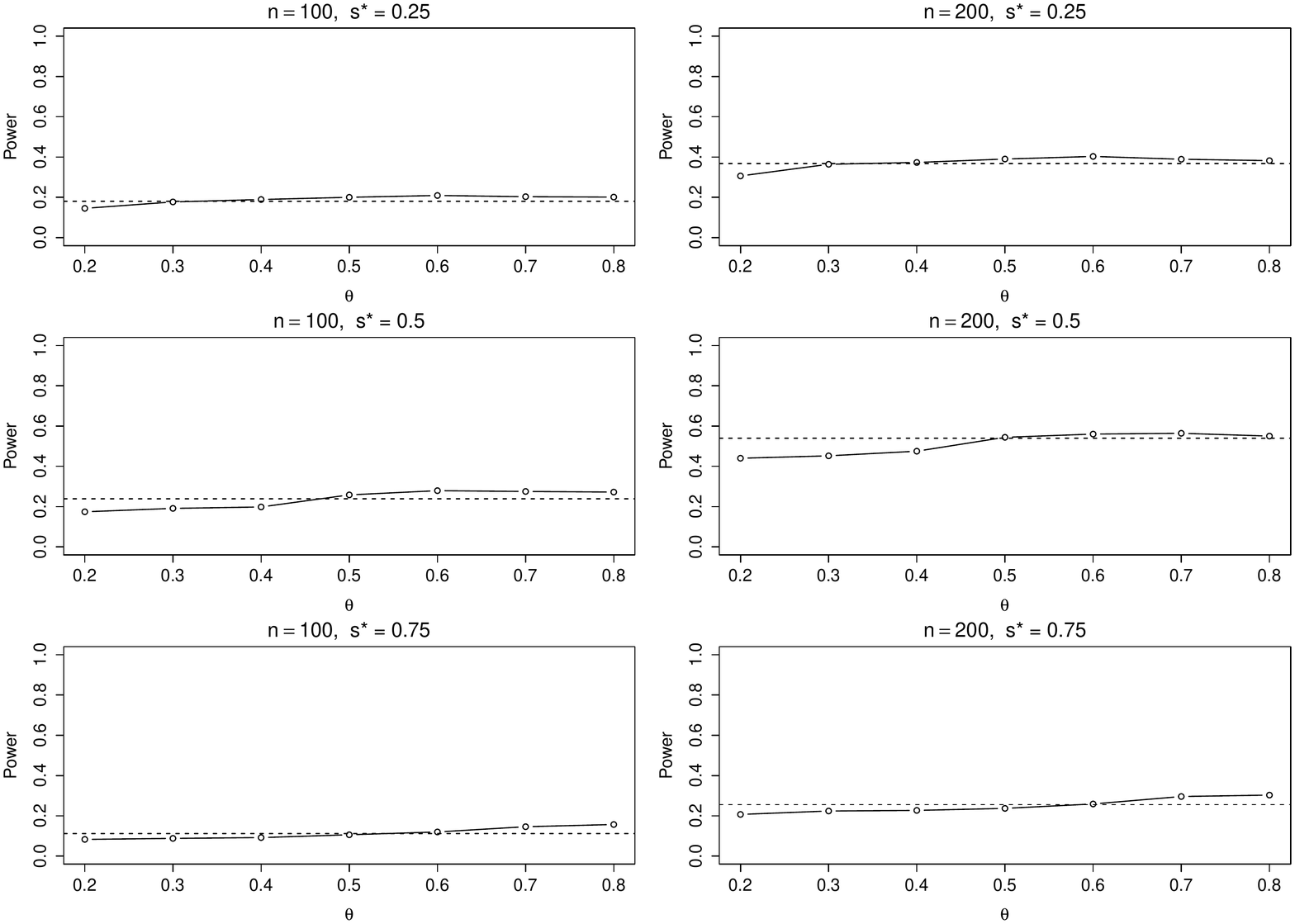}
\caption{Rejection rates of the test based on $S_{n,A}^\theta$ ($\circ$) against $\theta \in \{0.2,0.3,\dots,0.8\}$ estimated from 1000 bivariate samples of size $n \in \{100,200\}$, such that, for each sample, the first $\ip{n s^\star}$ (resp.\ last $n - \ip{n s^\star}$) observations are generated from a Gumbel--Hougaard copula with parameter~2 (resp. 3). The dashed line marks the corresponding estimated rejection rate of the test based on $S_{n,A}$.}
\label{fig:powertheta}
\end{figure}

\vspace{-.3cm}
\paragraph{Empirical power of the test based on $S_{n,A}^\theta$} As a last experiment, we investigated the influence of the value $\theta$ on the power of the test based  on $S_{n,A}^\theta$. Figure~\ref{fig:powertheta} displays the rejection rates of the test based on $S_{n,A}^\theta$ against $\theta \in \{0.2,0.3,\dots,0.8\}$ estimated from 1000 bivariate samples of size $n \in \{100,200\}$, such that, for each sample, the first $\ip{n s^\star}$ (resp.\ last $n - \ip{n s^\star}$) observations are generated from a Gumbel--Hougaard copula with parameter 2 (resp. 3). The values 0.25, 0.5 and 0.75 were considered for $s^\star$. As one can see, the rejections rates are not too much affected by the value of~$\theta$. In addition, the power of the test based on $S_{n,A}^\theta$ remains overall reasonably close to that of the test based on $S_{n,A}$. From a practical perspective, the latter result suggests that, under $\Hc_{0,m}$ in~\eqref{eq:H0m}, the somehow ``non optimal'' use of the test based on $S_{n,A}^\theta$ instead that based on $S_{n,A}$ does not incur a large power loss, if any. As a consequence, if one hesitates about which of $\Hc_{0,m}$ in~\eqref{eq:H0m} or $\Hc_{1,m}$ in~\eqref{eq:H1m} holds, it seems safer to use the test based on $S_{n,A}^\theta$ as, should $\Hc_{1,m}$ be actually true, the latter test is more likely to hold its level by construction, and should $\Hc_{0,m}$ be true, the power loss, if any, should not be too large.

\section{Illustrations}\label{sec:illustration}
 \def\theequation{6.\arabic{equation}}
\setcounter{equation}{0}

\begin{figure}
\centering
\includegraphics[width=0.8\linewidth]{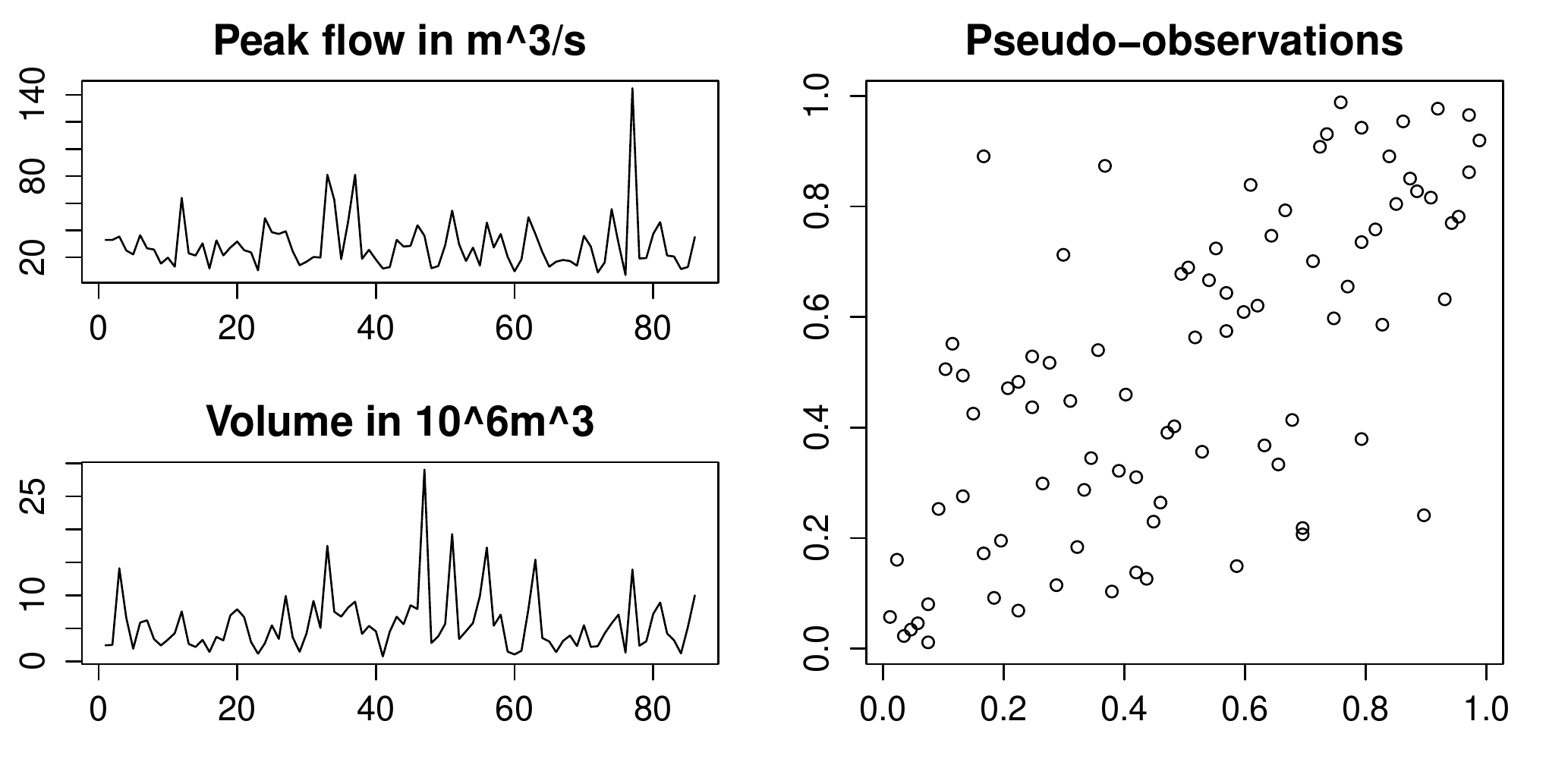}
\caption{Left: Annual maximal peak flows and volumes of discharges measured in Streckewalde, Germany. Right: Corresponding pseudo-observations computed using~\eqref{eq:pseudo1n}.}
\label{fig:peakvolume}
\end{figure}

To illustrate the proposed tests, we consider two hydrological data sets. The first one consists of $n=86$ bivariate annual maxima measured between 1921 and 2011 (five years of data are missing) at a station located on the river Pre\ss nitz in Streckewalde, Germany. The variables of interest are $Q$, the annual maximal peak flow (in $m^3/s$), and $V$, the annual maximal volume of discharge (in $10^6 \, m^3$). Their observations are represented in Figure~\ref{fig:peakvolume}. The joint distribution of $Q$ and $V$ is of strong interest to hydrologists as it can be used to assess the risk of catastrophic flood levels. For a recent case study, we refer to \cite{BacHal14}. 

Because we are dealing with bivariate block maxima, it is natural to assume that the data arise from one or more bivariate extreme-value distributions. 
The aim of our analysis is to test for possible changes in the dependence between $Q$ and $V$ that might have occurred during the long period of observation. An additional element to be taken into account here is that a dam was built on the river Pre\ss nitz in 1973 (which corresponds to the $48$th observation) a few kilometers upstream from the measurement station. We make the hypothesis that, if there are changes in the two components series, then, they are unique and they occurred simultaneously after observation~48 due to the construction of the dam. In other words, we assume that either $\Hc_{0,m}$ in~\eqref{eq:H0m} holds or $\Hc_{1,m}$ in~\eqref{eq:H1m} with $\theta=48/86$ holds. In the former case, it is natural to use the test for change-point detection based on $S_{n,A}$ in~\eqref{eq:SnA}, while in the latter case, the extension based on $S_{n,A}^\theta$ in~\eqref{eq:SnAtheta} with $\theta=48/86$ should be preferred. As mentioned in the previous section (see Figure~\ref{fig:powertheta} and the related discussion), using the test based on $S_{n,A}^{\scriptscriptstyle \theta}$ for some value of~$\theta$ when $\Hc_{0,m}$ in~\eqref{eq:H0m} actually holds does not seem to result in a strong power loss, if any. For that reason, we carried out the test based on $S_{n,A}^{\scriptscriptstyle \theta}$ with $\theta=48/86$. The resulting approximate p-value of 0.068, obtained from $B=10000$ multiplier bootstrap replicates, indicates that there is some weak evidence of change in the dependence between $Q$ and $V$. Interestingly enough, the maximum of the test statistic was not obtained for observation 48 but for observation 32 corresponding to year 1953. 

\begin{figure}
\centering
\includegraphics[width=0.8\linewidth]{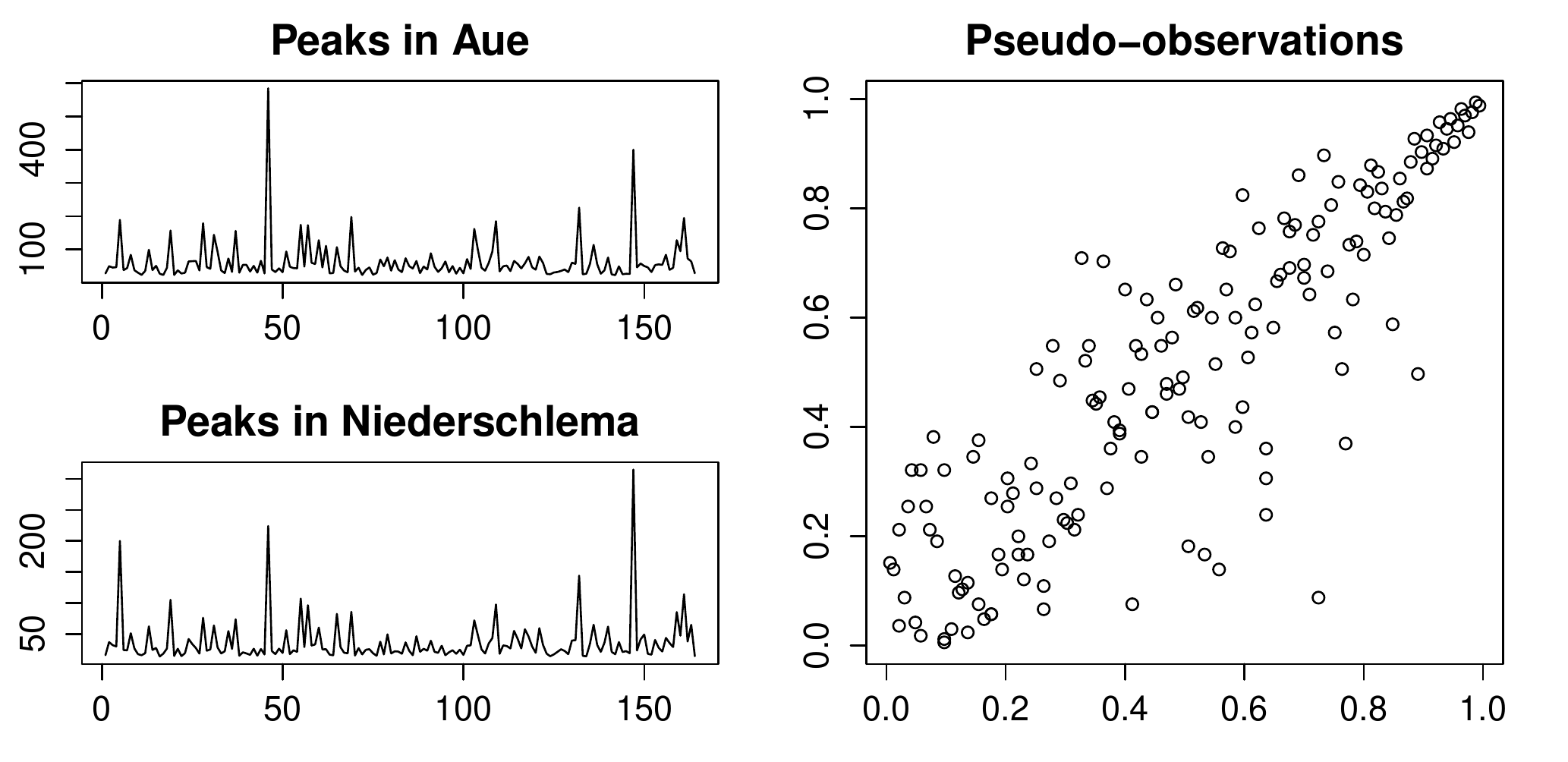}
\caption{Left: Peak flows in $m^3/s$ of 164 summer flood events simultaneously measured at gauges in Aue and Niederschlema, Germany. Right: Corresponding pseudo-observations computed using~\eqref{eq:pseudo1n}.}
\label{fig:peakpeak}
\end{figure}

The second data set consists of peak flows (in $m^3/s$) simultaneously measured at two neighboring gauges for $n=164$ physically independent summer flood events. The two gauges are located in Germany, in Aue and Niederschlema, respectively, and the corresponding measurements will thus be denoted by $Q_A$ and $Q_N$, respectively. The observations, chronologically ordered, span the period 1929-2011 and are represented in Figure~\ref{fig:peakpeak}. 
An event was classified as a flood, if each peak flow exceeded the smallest annual maximal peak flow measured between 1929 and 2011 in Aue and Niederschlema, respectively. The period of each flood event was identified by hand and only the largest value (peak flow) was included in the data set. Hence, by construction, the observations are formed subject to a block maximal procedure, with possibly slightly differing block sizes for each of the flood events. It therefore seems sensible to assume that the data-generating distribution(s) are extreme-value distributions.

 There were two reasons why only summer events were included in the analysis. First, typical winter floods are produced from melting snow, whereas summer floods are due to short but heavy rainfalls. These very different physical mechanisms lead to different peak flow distributions. Second, very high peak flows, which are of particular interest, almost exclusively occur during the summer time. 
The joint distribution of peak flows is of interest, for instance, to evaluate the efficiency of water reservoirs \citep{SchSch15}.

The aim of our analysis is to assess whether the dependence between $Q_A$ and $Q_N$ changed during the long observation period. As for the previous illustration, it might be important to take into account the fact that dams where constructed on the river Mulde and one of its tributary upstream of the two gauges Aue and Niederschlema. A first dam, called Sch\"onheiderhammer, was put in service in 1980 (which corresponds to observation 108) and a second dam, named Eibenstock, was put into service in 1982. As previously, we make the hypothesis that, if there are changes in the two components series, then, they are unique and they occurred simultaneously after observation~108 due to the construction of the dams. Following the same reasoning as for the first illustration, we apply the test based on $S_{n,A}^\theta$ in~\eqref{eq:SnAtheta} with $\theta=108/164$ and obtain an approximate p-value of 0.195 based on $B=10000$ multiplier bootstrap replicates. Hence, there is no evidence for a change in the dependence between $Q_A$ and $Q_N$.

\appendix

\section{Proofs of Propositions~\ref{prop:ferreira},~\ref{prop:weakdim2} and~\ref{prop:weakdim2break}} 
\label{sec:proofs}
\def\theequation{A.\arabic{equation}}
\setcounter{equation}{0}

Propositions~\ref{prop:ferreira} and~\ref{prop:weakdim2} will be corollaries of a more general result. Recall that $\Delta=\{ (s,s')\in[0,1]: s \le s'\}$ and let $\lambda_n(s,s')= (\ip{ns'} - \ip{ns})/n$ for $(s,s') \in \Delta$. Also, consider the process
\begin{equation}
\label{eq:Abn}
\Ab_n(s,s',t) =  \sqrt n \lambda_n(s,s') \{ \hat A_{\ip{ns}+1:\ip{ns'}}(t) - A(t)\}, \quad (s,s',t) \in \Delta \times [0,1],
\end{equation}
where $\hat A_{\ip{ns}+1:\ip{ns'}}$ is defined in~\eqref{eq:hatAkl}, and note that 
\begin{equation}
\label{eq:Abn01}
\Ab_n(0,1,t) = \sqrt n \{ \hat A_{1:n}(t) - A(t) \}, \quad t \in [0,1]
\end{equation}
is the process of interest in Proposition~\ref{prop:ferreira}.

\begin{theorem} \label{theo:twosided}
Under the conditions of Proposition~\ref{prop:weakdim2}, in the normed space $(\ell^\infty(\Delta \times [0,1]), \| \cdot \|_\infty)$, we have
$\Ab_n  \weak \Ab_C$, where
\begin{equation}
\label{eq:AbC}
\Ab_C(s,s',t) = - \{1+A(t) \}^2  \int_0^1 \Cb_C(s,s', y^{1-t},y^t) \, dy.
\end{equation}
\end{theorem}

\begin{proof}
Since $\int_0^1\ind(m\leq y)dy=(1-m)$ for $0\leq m\leq1$, we can write, for $1\le k \leq \ell \le n$,
\begin{align*}
\nonumber
\hat S_{k:\ell}(t) 
&= 
1 - \int_{0}^1 \frac{1}{\ell-k+1} \sum_{i=k}^\ell \ind \big\{  \max  ( \hat U_{k:\ell,i}^{1/(1-t)},   \hat V_{k:\ell,i}^{1/t} ) \le y \big\} \, dy \\
\nonumber
&= 
1 - \int_{0}^1 \frac{1}{\ell-k+1} \sum_{i=k}^\ell \ind \big( \hat U_{k:\ell,i} \le y^{1-t},   \hat V_{k:\ell,i} \le y^t \big) \, dy \\
&=
1- \int_{0}^1 C_{k:\ell} (y^{1-t}, y^t) \, dy,
\end{align*}
where $C_{k:\ell}$ denotes the empirical copula, see~\eqref{eq:empcop}. Similarly, for $t \in [0,1]$,
\[
S(t) = 1-\{1+A(t)\}^{-1}= 1 - \int_0^1 y^{A(t)} \, dy  = 1- \int_0^1 C(y^{1-t}, y^t) \,dy.
\]
Therefore, introducing the notation 
\[
B_n(s,s',t)= \int_0^1 C_{\ip{ns}+1:\ip{ns'} } (y^{1-t}, y^t) \, dy \times \int_0^1 C (y^{1-t}, y^t)  \, dy,
\] 
we obtain, for any $(s,s') \in \Delta$ such that $\ip{ns} < \ip{ns'}$,
\begin{align} \label{eq:AbnCbn}
\Ab_n(s,s',t) 
&=
 \sqrt n \lambda_n(s,s') \left\{\frac{\hat S_{\ip{ns}+1:\ip{ns'}}(t)}{1- \hat S_{\ip{ns}+1:\ip{ns'}}(t)}- \frac{S(t)}{1-S(t)}  \right\} \nonumber \\
&= 
\frac{\sqrt n \lambda_{n}(s,s') }{B_n(s,s',t) } \int_0^1 C(y^{1-t},y^t) -  C_{\ip{ns}+1:\ip{ns'} } (y^{1-t}, y^t) \, dy \nonumber \\ 
&= 
- \frac{1}{B_n(s,s', t) } \int_0^1 \Cb_n(s,s',y^{1-t}, y^t)\, dy ,
\end{align}
where, for any $(s,s',u,v) \in \Delta \times [0,1]^2$,
\begin{equation*} 
\Cb_n(s,s',u,v)
=
\sqrt n \lambda_n(s,s') \left\{ C_{\ip{ns}+1_\ip{ns'}} (u,v) - C(u,v)\right\}
\end{equation*}
denotes the {\em two-sided sequential empirical copula process} studied in \cite{BucKoj14}. Since $A$ is continuously differentiable on $(0,1)$, we have, from Example~5.3 in \cite{Seg12}, that the first-order partial derivatives $\dot C_1$ and $\dot C_2$ of $C$ exist and are continuous on $(0,1) \times [0,1]$ and $[0,1] \times (0,1)$, respectively. Hence, from Theorem 3.4 in \cite{BucKoj14},
\begin{equation}
\label{eq:aeCbn}
\sup_{(s,s',u,v) \in \Delta \times [0,1]^2} | \Cb_n(s,s',u,v) - \bar \Cb_n(s,s',u,v) | = o_\Prob(1),
\end{equation}
where, for any $(s,s',u,v) \in \Delta \times [0,1]^2$,
\begin{equation}
\label{eq:barCbn}
\bar \Cb_n(s,t,u,v)  = \bar \Bb_n(s,s',u,v)  - \dot C_1(u,v) \bar \Bb_n(s,s',u,1) - \dot C_2(u,v) \bar \Bb_n(s,s',1,v),
\end{equation}
with the convention that $\dot C_1(x,\cdot) = \dot C_2(\cdot,x) = 0$ if $x \in \{0,1\}$, and where
$$
\bar \Bb_n(s,s',u,v) = \frac{1}{\sqrt{n}} \sum_{i=\ip{ns}+1}^{\ip{ns'}} \{ \ind \big( U_i \le u, V_i \le v \big) - C(u,v)\Big \}
$$
with $(U_i,V_i) = (F(X_i),G(Y_i))$, $i=1,\dots,n$.
From~\eqref{eq:aeCbn}, the fact that $\bar \Bb_n(0,\cdot,\cdot,\cdot) \weak \Bb_C$ in $(\ell^\infty([0,1]^3),\|\cdot\|_\infty)$ under $\Hc_{0,c}$ \citep[see, e.g.,][Theorem 2.12.1]{VanWel96}, the fact that $|\dot C_1| \leq 1$ and $|\dot C_2| \leq 1$ and the continuous mapping theorem, it follows that the process
\begin{align} \label{eq:tildeAbn}
\tilde \Ab_n(s,s',t) 
&= - \{1+A(t)\}^2 \times \int_0^1 \Cb_n(s,s',y^{1-t}, y^t)\, dy 
\end{align}
weakly converges in $(\ell^\infty(\Delta \times [0,1]), \|\cdot\|_\infty)$ to $\Ab_C$ in~\eqref{eq:AbC}. The theorem is thus proved if we show that
\begin{align}  \label{eq:sup}
\sup_{(s,s',t) \in\Delta \times [0,1]}|\Ab_n(s,s',t) - \tilde \Ab_n(s,s',t)|
= o_\Prob(1).
\end{align}
For that purpose, let us first show  that, for any $(s,s') \in \Delta$ with $\ip{ns}<\ip{ns'}$ and any $t\in [0,1]$, 
\begin{align} \label{eq:bnsbound}
B_n(s,s',t) \ge \frac{1}{16}.
\end{align}
First of all, we have, for any $u,v \in [3/4,1]$ and any $1 \le k < \ell \le n$,
\begin{align*}
C_{k:\ell}(u,v) 
&\ge 
C_{k:\ell}(3/4,3/4)
= \frac{1}{\ell-k+1} \sum_{i=k}^\ell \ind\{ \hat U_{k:\ell,i} \le 3/4, \hat V_{k:\ell,i} \le 3/4\}.
\end{align*}
The number of $\hat U_{k:\ell,i}$ is  $(\ell-k+1)$, and of those, exactly $\ip{3(\ell-k+2)/4}$ do not exceed $3/4$. The same is true for the $\hat V_{k:\ell,i}$. Hence, for at least
\[
\ip{3(\ell-k+2)/4} - \left\{ (\ell-k+1) -  \ip{3(\ell-k+2)/4} \right\}
\]
of the pairs, both pseudo-observations do not exceed $3/4$. As a consequence,
\[
C_{k:\ell}(3/4, 3/4) 
\ge 
\frac{2 \ip{3(\ell-k+2)/4} - (\ell-k+1)}{\ell-k+1} 
\ge
\frac{3}{2} \frac{\ell-k+2}{\ell-k+1} - 1 \ge \frac{1}{2}.
\]
In particular, we have $C_{k:\ell}(y^{1-t}, y^t) \ge 1/2$ for all $y$ such that $\min(y^{1-t},y^t) \ge 3/4$. Since $\min(y^{1-t},y^t) \ge y$, we get that
\begin{equation}
\label{eq:boundCkl}
\int_0^1 C_{k:\ell} (y^{1-t}, y^t) \, dy 
\ge
\int_{3/4}^1 \frac{1}{2} \, dy = \frac{1}{8}.
\end{equation}
Together with the fact that $\int_0^1 C(y^{1-t}, y^t) \, dy = \{ 1+A(t) \}^{-1} \ge 1/2$, the last display implies the bound in~\eqref{eq:bnsbound}.

In order to prove~\eqref{eq:sup}, notice first that, by~\eqref{eq:AbnCbn},~\eqref{eq:tildeAbn}, the triangle inequality and~\eqref{eq:bnsbound},
$$
\sup_{(s,s',t) \in\Delta \times [0,1]}|\Ab_n(s,s',t) - \tilde \Ab_n(s,s',t)| \leq 20 \sup_{(s,s',u,v) \in\Delta \times [0,1]^2}|\Cb_n(s,s',u,v)|, 
$$
both sides of the inequality being zero if the suprema are restricted to $(s,s') \in \Delta$ such that $\ip{ns} = \ip{ns'}$. Next, fix $\eps, \eta > 0$. Using the fact that that $\Cb_n$ vanishes when $s=s'$ and that $\Cb_n$ is asymptotically uniformly equicontinuous in probability, there exists $\delta \in (0,1)$ such that, for all sufficiently large $n$,
\begin{multline*}
\Prob \left\{ \sup_{(s,s',t) \in\Delta \times [0,1] \atop s' - s < \delta}|\Ab_n(s,s',t) - \tilde \Ab_n(s,s',t)| > \eps \right\} \\ \leq \Prob \left\{ 20 \sup_{(s,s',u,v) \in\Delta \times [0,1]^2 \atop s' - s < \delta}|\Cb_n(s,s',u,v)| > \eps \right\} < \eta.
\end{multline*}
The proof of~\eqref{eq:sup} will thus be complete if we show that 
$$
\sup_{(s,s',t) \in\Delta \times [0,1] \atop s' - s \geq \delta}|\Ab_n(s,s',t) - \tilde \Ab_n(s,s',t)| = o_\Prob(1),
$$
which, in view of~\eqref{eq:bnsbound}, would be an immediate consequence of the fact that
\[
\sup_{(s,s',t) \in\Delta \times [0,1] \atop s' - s \geq \delta}|B_n(s,s',t) - \{ 1+A(t) \}^{-2}| = o_\Prob(1).
\]
Using again the identity $\{1+A(t)\}^{-1}=  \int_0^1 C(y^{1-t}, y^t) \,dy$, the last display follows from the fact that
\begin{multline*}
\sup_{(s,s',t) \in\Delta \times [0,1] \atop s' - s \geq \delta} \left| \int_0^1 C_{\ip{ns}+1: \ip{ns'} } (y^{1-t}, y^t) \, dy - \{ 1+A(t) \}^{-1} \right| \\
= n^{-1/2} \sup_{(s,s',t) \in\Delta \times [0,1] \atop s' - s \geq \delta} \frac{1}{\lambda_n(s,s')} \left| \int_0^1 \Cb_n (s,s',y^{1-t}, y^t) \, dy \right|
=o_\Prob(1),
\end{multline*}
which completes the proof.
\end{proof}

\begin{proof}[\bf Proof of Proposition~\ref{prop:ferreira}.]
The proposition immediately follows from~\eqref{eq:Abn01} and Theorem~\ref{theo:twosided}.
\end{proof}

\begin{proof}[\bf Proof of Proposition~\ref{prop:weakdim2}.] 
The assertion is a mere consequence of the fact that, under $\Hc_0$,
\begin{align} \label{eq:Dbn}
\Db_n(s,t) 
&= 
\lambda_n(s,1) \Ab_n(0,s,t) - \lambda_n(0,s) \Ab_n(s,1,t), \quad (s,t) \in [0,1]^2,
\end{align}
Theorem~\ref{theo:twosided}, and the continuous mapping theorem.
\end{proof}

For the proof of Proposition~\ref{prop:weakdim2break}, another lemma, generalizing the identity in~\eqref{eq:aeCbn}, is needed.
\begin{lemma} \label{lem:aeCbnH1}
Under $\Hc_{1,m} \cap \Hc_{0,c}$ and if $A$ is continuously differentiable on $(0,1)$,
\begin{align*}
&\sup_{(s,s',u,v) \in ( \Delta \cap[0,\theta]^2 ) \times [0,1]^2} | \Cb_n(s,s',u,v) - \bar \Cb_n(s,s',u,v) | = o_\Prob(1), \\
&\sup_{(s,s',u,v) \in ( \Delta\cap[\theta,1]^2 ) \times [0,1]^2} | \Cb_n(s,s',u,v) - \bar \Cb_n(s,s',u,v) | = o_\Prob(1).
\end{align*}
\end{lemma}

\begin{proof}
Both suprema are measurable, and they are equal in distribution to the same suprema calculated under $\Hc_{0,m}\cap \Hc_{0,c}$. The assertions then follow from \eqref{eq:aeCbn}.
\end{proof}

\begin{proof}[\bf Proof of Proposition~\ref{prop:weakdim2break}]
Let
\begin{equation} \label{eq:Abntheta}
\Ab_n^\theta(s,s',t) =  \sqrt n \lambda_n(s,s') \{ \hat A_{\ip{ns}+1:\ip{ns'}}^\theta(t) - A(t)\}, \quad (s,s',t) \in \Delta \times [0,1],
\end{equation}
where $\hat A_{\ip{ns}+1:\ip{ns'}}^\theta$ is defined in~\eqref{eq:Aklbreak}. We shall first show that
\begin{equation} \label{eq:ae1}
\sup_{(s,t) \in [0,\theta] \times [0,1]} | \Ab_n^\theta(s,1,t) - \tilde \Ab_n(s,\theta,t) - \tilde \Ab_n(\theta,1,t) | = o_\Prob(1)
\end{equation}
and that 
\begin{equation} \label{eq:ae2}
\sup_{(s,t) \in [\theta,1] \times [0,1]} | \Ab_n^\theta(0,s,t) - \tilde \Ab_n(0,\theta,t) - \tilde \Ab_n(\theta,s,t) | = o_\Prob(1),
\end{equation}
where $\tilde \Ab_n$ is defined in~\eqref{eq:tildeAbn}. Proceeding as in~\eqref{eq:AbnCbn}, for $(s,t) \in [0,\theta] \times [0,1]$, we have
\begin{align*}
\Ab_n^\theta(s,1,t) &= \sqrt{n} \lambda_n(s,1) \frac{\hat S_{\ip{ns}+1:n}^\theta(t) - S(t)}{\{1-\hat S_{\ip{ns}+1:n}^\theta(t) \} \{1 - S(t) \}} \\
&=  - \frac{\int_0^1 \Cb_n(s,\theta,y^{1-t},y^t) + \Cb_n(\theta,1,y^{1-t},y^t) \, dy}{\{1-\hat S_{\ip{ns}+1:n}^\theta(t) \} \{1 - S(t) \}}.
\end{align*}
Using the fact that $3/2 \leq \{1 - S(t) \}^{-1} = \{ 1 + A(t) \} \leq 2$ for all $t \in [0,1]$, the supremum on the left of~\eqref{eq:ae1} is smaller than $2 I_{1,n} \times I_{2,n}$, where
$$
I_{1,n} = \sup_{(s,t) \in [0,\theta] \times [0,1]} \left| \int_0^1 \Cb_n(s,\theta,y^{1-t},y^t) + \Cb_n(\theta,1,y^{1-t},y^t) \, dy \right| 
$$
and 
$$
I_{2,n} = \sup_{(s,t) \in [0,\theta] \times [0,1]} \left| \{ 1 - \hat S_{\ip{ns}+1:n}^\theta(t) \}^{-1} - \{ 1 - S(t) \}^{-1} \right|.
$$
From Lemma~\ref{lem:aeCbnH1}, the weak convergence of $\bar \Cb_n \weak \Cb_C$ in $(\ell^\infty(\Delta \times [0,1]^2), \| \cdot \|_\infty)$ under $\Hc_{0,c}$, and the continuous mapping theorem, $I_{1,n}= O_\Prob(1)$. Concerning $I_{2,n}$, by the definition of $S_{k:\ell}^\theta$ in~\eqref{eq:Sklbreak}, we have that
$$
 \sup_{(s,t) \in [0,\theta] \times [0,1]} \left| \hat S_{\ip{ns}+1:n}^\theta(t) - S(t) \right| 
 = I_{1,n} \times n^{-1/2} \sup_{s \in [0,\theta]} \{\lambda_n(s,1) \}^{-1} = o_\Prob(1).
$$
Hence, $\{ (s,t) \mapsto \hat S_{\ip{ns}+1:n}^\theta(t) \} \p \{ (s,t) \mapsto S(t) \}$ in $(\ell^\infty([0,\theta] \times [0,1]), \| \cdot \|_\infty)$, which, from the continuous mapping theorem, implies that $I_{2,n} = o_\Prob(1)$. This completes the proof of~\eqref{eq:ae1}. The proof of~\eqref{eq:ae2} is similar.
  
To finish the proof, notice first that, under $\Hc_{0,c} \cap \Hc_{1,m}$, 
$$
\Db_n^\theta(s,t) = \lambda_n(s,1) \Ab_n^\theta(0,s,t) - \lambda_n(0,s) \Ab_n^\theta(s,1,t), \quad (s,t) \in [0,1]^2,
$$
where $\Ab_n^\theta$ is defined in~\eqref{eq:Abntheta}. Next, let 
\begin{multline*}
\bar \Db_n^\theta(s,t) = \lambda_n(s,1) \{ \bar \Ab_n(0,s \wedge \theta,t) + \bar \Ab_n(s \wedge \theta,s,t)\} \\ - \lambda_n(0,s) \{ \bar \Ab_n(s,s \vee \theta,t) + \bar \Ab_n(s \vee \theta,1,t) \}, \quad (s,t) \in [0,1]^2,
\end{multline*}
where $\wedge$ and $\vee$ denote the minimum and maximum operators, respectively,  
\begin{equation} \label{eq:barAbn}
\bar \Ab_n(s,s',t) = - \{1+A(t)\}^2 \times \int_0^1 \bar \Cb_n(s,s',y^{1-t}, y^t)\, dy, 
\end{equation}
and $\bar \Cb_n$ is defined in~\eqref{eq:barCbn}, and let
\begin{multline*}
\tilde \Db_n^\theta(s,t) = \lambda_n(s,1) \{ \tilde \Ab_n(0,s \wedge \theta,t) + \tilde \Ab_n(s \wedge \theta,s,t)\} \\ - \lambda_n(0,s) \{ \tilde \Ab_n(s,s \vee \theta,t) + \tilde \Ab_n(s \vee \theta,1,t) \}, \quad (s,t) \in [0,1]^2.
\end{multline*} 
The desired result shall then follow from the fact that
\begin{align}
\label{eq:aetildeDntheta}
\sup_{(s,t) \in [0,1]^2} | \Db_n^\theta(s,t) - \tilde \Db_n^\theta(s,t) | = o_\Prob(1), \\
\label{eq:aebarDntheta}
\sup_{(s,t) \in [0,1]^2} | \tilde \Db_n^\theta(s,t) - \bar \Db_n^\theta(s,t) | = o_\Prob(1),
\end{align}
and
\begin{equation}
\label{eq:weakbarDntheta}
\bar \Db_n^\theta \weak \Db_C \quad \mbox{in} \quad (\ell^\infty([0,1]^2),\|\cdot\|_\infty),
\end{equation}
where $\Db_C$ is given in~\eqref{eq:DbC}. To show that~\eqref{eq:aetildeDntheta} holds, it suffices to restrict the supremum in~\eqref{eq:aetildeDntheta} successively to $(s,t) \in [0,\theta] \times [0,1]$ and $(s,t) \in [\theta,1] \times [0,1]$ and use the triangle inequality,~\eqref{eq:sup}, and~\eqref{eq:ae1} and~\eqref{eq:ae2}, respectively. The fact that~\eqref{eq:aebarDntheta} holds is obtained from the triangle inequality and Lemma~\ref{lem:aeCbnH1}. Finally,~\eqref{eq:weakbarDntheta} follows from the fact that, for any $0 \leq s \leq s' \leq s'' \leq 1$, $\bar \Ab_n(s,s'',\cdot) = \bar \Ab_n(s,s',\cdot) + \bar \Ab_n(s',s'',\cdot)$, the weak convergence of $\bar \Cb_n$ under $\Hc_{0,c}$ and the continuous mapping theorem.
\end{proof}

\section{Proofs of Propositions~\ref{prop:bootdim2} and~\ref{prop:bootdim2break}}
\label{sec:proofsboot}
\def\theequation{B.\arabic{equation}}
\setcounter{equation}{0}

Just as for the non-bootstrap results in Propositions~\ref{prop:ferreira} and~\ref{prop:weakdim2}, Propositions~\ref{prop:bootdim2} and~\ref{prop:bootdim2break} can be conveniently proved using appropriate two-sided sequential processes. For $(s,s',u,v) \in \Delta\times [0,1]^2$ and $b=1,\dots,B$, let
\begin{equation}
\label{eq:barBnb}
\bar \Bb_n^{(b)}(s,s',u,v) = \frac{1}{\sqrt{n}} \sum_{i=\ip{ns}+1}^{\ip{ns'}} \xi_i^{(b)}\Big\{ \ind \big( U_i \le u, V_i \le v \big) - C(u,v)\Big \},
\end{equation}
and
\begin{align}
\label{eq:tildeBbnb}
\tilde \Bb_n^{(b)}(s, s',u,v)
=
\frac{1}{\sqrt{n}}\sum_{i=\ip{ns}+1}^{\ip{ns'}}\xi_i^{(b)} \Big\{\ind \big( \hat{U}_{\ip{ns}+1:\ip{ns'},i} \le u, \hat{V}_{\ip{ns}+1:\ip{ns'},i} \le v \big) 
- C(u,v)\Big\}.
\end{align}
Next, for $(s,s',u,v) \in \Delta \times [0,1]^2$ and $b=1,\dots,B$, let
\begin{align*}
\bar \Cb_n^{(b)}(s,s',u,v) = \bar \Bb_n^{(b)}(s,s',u,v)  - \dot C_1(u,v) \bar \Bb_n^{(b)}(s,s',y^{1-t},1) - \dot C_2(u,v) \bar \Bb_n^{(b)}(s,s',1,y^t),
\end{align*}
and
$$
\tilde \Cb_n^{(b)}(s,s',u,v) = \tilde \Bb_n^{(b)}(s,s',u,v) - \dot C_1(u,v) \tilde \Bb_n^{(b)}(s,s',u,1) - \dot C_2(u,v) \tilde \Bb_n^{(b)}(s,s',1,v).
$$
Furthermore, from~\eqref{eq:evc}, for $(y,t) \in (0,1) \times [0,1]$, we obtain that
\begin{equation}
\label{eq:dotc1}
\dot{C}_1(y^{1-t},y^t)=\left\{ A(t)-tA'(t)\right \} y^{A(t)-(1-t)}
\end{equation}
and
\begin{equation}
\label{eq:dotc2}
\dot{C}_2(y^{1-t},y^t)=\left \{ A(t)+(1-t)A'(t)\right \}y^{A(t)-t},
\end{equation}
where $A'$ is extended by continuity at 0 and 1. Indeed, from the fact that $\max(t,1-t) \le A(t) \le 1$ for all $t \in [0,1]$, we have that $A'(t) \in [-1,1]$ for all $t \in (0,1)$, and, from the convexity of $A$ on $[0,1]$, we have that $A'$ is increasing on $(0,1)$. In addition, we adopt the usual convention that $0^0 = 1$. Finally, for $b=1,\dots,B$ and $(s,s',t) \in \Delta \times [0,1]$, let
\begin{equation}
\label{eq:barAbnb}
\bar \Ab_n^{(b)}(s,s',t) = - \{1 + A(t)\}^2 \times \int_0^1 \bar \Cb_n^{(b)}(s,s',y^{1-t},y^t) \, dy.
\end{equation}
and
\begin{equation}
\label{eq:tildeAbnb}
\tilde \Ab_n^{(b)}(s,s',t) = - \{1 + A(t)\}^2 \times \int_0^1 \tilde \Cb_n^{(b)}(s,s',y^{1-t},y^t) \, dy.
\end{equation}

\begin{lemma}\label{lem:multbarAn}
Under $\Hc_{0,c}$  and if $A$ is continuously differentiable on $(0,1)$,
$$
\Big( \bar \Ab_n, \bar \Ab_n^{(1)},\dots, \bar \Ab_n^{(B)} \Big) \weak \Big( \Ab_C,\Ab_C^{(1)},\dots,\Ab_C^{(B)} \Big)
$$
in $(\ell^\infty(\Delta \times [0,1]), \| \cdot \|_\infty)^{B+1}$, where $\bar \Ab_n$ is defined in~\eqref{eq:barAbn}, $\Ab_C$ is defined in~\eqref{eq:AbC} and $\Ab_C^{(1)},\dots,\Ab_C^{(B)}$ are independent copies of $\Ab_C$.
\end{lemma}

\begin{proof}
From Theorem~2.1 of \cite{BucKoj14}, the fact that $|\dot C_1| \leq 1$ and $|\dot C_2| \leq 1$, the continuous mapping theorem and~\eqref{eq:aeCbn}, we have that, under $\Hc_{0,c}$,
$$
\Big(\bar \Cb_n, \bar \Cb_n^{(1)},\dots, \bar \Cb_n^{(B)} \Big) \weak \Big( \Cb_C,\Cb_C^{(1)},\dots,\Cb_C^{(B)} \Big),
$$
in $(\ell^\infty(\Delta \times [0,1]^2), \| \cdot \|_\infty)^{B+1}$, where $\bar \Cb_n$ and $\Cb_C$ are defined in~\eqref{eq:barCbn} and Proposition~\ref{prop:weakdim2}, respectively, and  $\Cb_C^{(1)},\dots,\Cb_C^{(B)}$ are independent copies of~$\Cb_C$. The desired follows from the continuous mapping theorem.
\end{proof}

\begin{lemma}\label{lem:aetildeAnb}
Under $\Hc_0$  and if $A$ is continuously differentiable on $(0,1)$, for any $b=1,\dots,B$,
$$
\sup_{(s,s',t) \in \Delta \times [0,1]} | \tilde \Ab_n^{(b)}(s,s',t) - \bar \Ab_n^{(b)}(s,s',t) | = o_\Prob(1).
$$
\end{lemma}

\begin{proof}
From the proof of Proposition~4.3 of \cite{BucKojRohSeg14} (see the term (B.3)), we have that, for $b=1,\dots,B$,
$$
\sup_{(s,s',u,v) \in \Delta \times [0,1]^2} | \check \Bb_n^{(b)}(s,s',u,v) - \bar \Bb_n^{(b)}(s,s',u,v) | = o_\Prob(1),
$$
where $\check \Bb_n^{(b)}$ and $ \bar \Bb_n^{(b)}$ are defined in~\eqref{eq:checkBnb} and~\eqref{eq:barBnb}. The desired result will follow if we show that 
\begin{equation}
\label{eq:aebarBnb}
\sup_{(s,s',u,v) \in \Delta \times [0,1]^2} | \tilde \Bb_n^{(b)}(s,s',u,v) - \check \Bb_n^{(b)}(s,s',u,v) | = o_\Prob(1),
\end{equation}
where $\tilde \Bb_n^{(b)}$ is defined in~\eqref{eq:tildeBbnb}. The supremum on the left of the previous display is smaller than
$$
\sup_{(s,s',u,v) \in \Delta \times [0,1]^2} | C_{\ip{ns}+1:\ip{ns'}}(u,v) - C(u,v) | \times \sup_{(s,s') \in \Delta} \left| \frac{1}{\sqrt{n}} \sum_{i=\ip{ns}+1}^{\ip{ns'}} \xi_i^{(b)} \right|.
$$
Using the fact that the first supremum is bounded by 2, the asymptotic uniform equicontinuity in probability of the process $(s,s') \mapsto n^{-1/2} \sum_{i=\ip{ns}+1}^{\ip{ns'}} \xi_i^{(b)}$ (by Donsker's theorem) which vanishes when $s=s'$, and the weak convergence $\Cb_n \weak \Cb_C$ in $(\ell^\infty(\Delta \times [0,1]^2), \| \cdot \|_\infty)$, we obtain that the latter display is $o_\Prob(1)$. Hence,~\eqref{eq:aebarBnb} holds. 
\end{proof}

For $(s,s',t) \in \Delta \times [0,1]$ and for $b=1, \dots, B$, let
\begin{equation}
\label{eq:checkAbn}
\check \Ab_n^{(b)}(s,s',t) 
=  
- \{ 1+ \hat A_{1:n}(t)\}^2  \times \frac{1}{\sqrt n} \sum_{i=\ip{ns}+1}^{\ip{ns'}} \xi_i^{(b)} \hat w_{\ip{ns}+1: \ip{ns'},i} (t),
\end{equation}
where $\hat w_{\ip{ns}+1: \ip{ns'},i}$ is defined in~\eqref{eq:wkl}. Furthermore, for $1\le k \le \ell \le n$ and $(y,t) \in (0,1) \times [0,1]$,  let
\begin{align}
\dot{C}_{1,k:\ell,n} (y^{1-t},y^t)&=
\left\{ \hat A_{k:\ell}(t )-t\hat{A}'_{k:\ell,n}(t)\right \} y^{\hat A_{k:\ell}(t)-(1-t)} 
,  \label{eq:estdotc1} \\ 
\dot{C}_{2,k:\ell,n} (y^{1-t},y^t)&=\left\{ \hat A_{k:\ell} (t)+(1-t)\hat{A}'_{k:\ell,n}(t)\right\}y^{\hat A_{k:\ell}(t)-t}, \label{eq:estdotc2} 
\end{align}
where $\hat{A}'_{k:\ell,n}$ is defined in~\eqref{eq:hatAkl'} and with the convention that $\dot{C}_{1,k:\ell,n} (x,\cdot) = \dot{C}_{2,k:\ell,n} (\cdot,x) = 0$ if $x \in \{0,1\}$. Finally, for $(s,s',y,t) \in \Delta \times [0,1]^2$ and $b=1,\dots,B$, let
\begin{align*}
\check \Cb_n^{(b)}(s,s',y^{1-t},y^t) 
= 
\check\Bb_n^{(b)}(s,s',y^{1-t},y^t)
&- \dot{C}_{1,\ip{ns}+1:\ip{ns'},n}(y^{1-t},y^t) \check \Bb_n^{(b)}(s,s',y^{1-t},1) \\
& - \dot{C}_{2,\ip{ns}+1:\ip{ns'},n}(y^{1-t},y^t) \check \Bb_n^{(b)}(s,s',1,y^t).
\end{align*}

\begin{lemma} \label{lem:exprcheckAbnb}
Let the pseudo-observations be either calculated as in~\eqref{eq:pseudoobs} or in~\eqref{eq:pseudotheta}. For $(s,s',t) \in \Delta \times [0,1]$ and $b=1,\dots,B$,
\begin{align} \label{eq:checkAbnCbn}
\check \Ab_n^{(b)}(s,s',t)
= -\{ 1+A_{1:n}(t)\}^2 \int_0^1 \check \Cb_n^{(b)}(s,s',y^{1-t},y^t) \, dy.
\end{align}
\end{lemma}

\begin{proof}
From the definitions of $\hat a_n, \hat b_n, \hat c_n$ and $\hat d_n$ given in Section~\ref{sec:testdim2}, we have
\begin{align*}
&\hspace{-.8cm} \int_0^1\check \Cb_n^{(b)}(s,s',y^{1-t},y^{t}) \, dy = \int_0^1 \check \Bb_n^{(b)}(s,s',y^{1-t},y^t) \, dy \nonumber \\
& \hspace{1cm}- \hat a_{\ip{ns}+1:\ip{ns'}}(t)\int_0^1y^{\hat b_{\ip{ns}+1:\ip{ns'}}(t)-1} \check\Bb_n^{(b)}(s,s',y^{1-t},1) \, dy \nonumber \\
&\hspace{1cm} - \hat c_{\ip{ns}+1:\ip{ns'}}(t)\int_0^1y^{\hat d_{\ip{ns}+1:\ip{ns'}}(t)-1} \check \Bb_n^{(b)}(s,s',1,y^t) \, dy \nonumber \\
&= (I_1-I_2-I_3)(s,s',t) , \nonumber
\end{align*}
where $I_1, I_2$ and $I_3$ are defined in an obvious manner. Observe that we can write
\begin{multline*}
\check \Bb_n^{(b)}(s,s',y^{1-t},y^t)
=
\frac{1}{\sqrt{n}}\sum_{i=\ip{ns}+1}^{\ip{ns'}}\xi_i^{(b)}\Big[\ind\left\{ \hat m_{\ip{ns}+1:\ip{ns'}, i}(t)\leq y\right \} \\
- \frac{1}{\ip{ns'} - \ip{ns}}\sum_{j=\ip{ns}+1}^{\ip{ns'}}\ind\left \{ \hat m_{\ip{ns}+1:\ip{ns'},j}(t)\leq y\right \} \Big].
\end{multline*}
Then, since $\int_0^1\ind(m\leq y)dy=(1-m)$ for $0\leq m\leq1$, we obtain that
\begin{align*}
I_1(s,s',y)
=
\frac{1}{\sqrt{n}}\sum_{i=\ip{ns}+1}^{\ip{ns'}}\xi_i^{(b)}\left \{ \widebar{m}_{\ip{ns}+1:\ip{ns'}}(t)-\hat m_{\ip{ns}+1:\ip{ns'},i}(t)\right \}.
\end{align*}
Similarly, since $\int_0^1 y^{b-1} \ind(m \le y) \, dy = (1- m^{b})/b$, 
\begin{align*}
I_2(s,s',y)
=
\frac{\hat a_{\ip{ns}+1:\ip{ns'}}(t)}{\hat b_{\ip{ns}+1:\ip{ns'}}(t) }\frac{1}{\sqrt{n}}\sum_{i=\ip{ns}+1}^{\ip{ns'}}\xi_i^{(b)}\left \{ \widebar{u}_{\ip{ns}+1:\ip{ns'}}(t)-\hat u_{\ip{ns}+1:\ip{ns'},i}(t)\right\} 
\end{align*}
and 
\begin{align*}
I_3(s,s',y)
=
\frac{\hat c_{\ip{ns}+1:\ip{ns'}}(t)}{\hat d_{\ip{ns}+1:\ip{ns'}}(t) } \frac{1}{\sqrt{n}}\sum_{i=\ip{ns}+1}^{\ip{ns'}}\xi_i^{(b)}\left\{ \widebar{v}_{\ip{ns}+1:\ip{ns'}}(t)-\hat v_{\ip{ns}+1:\ip{ns'},i}(t)\right \}.
\end{align*}
Hence,~\eqref{eq:checkAbnCbn} is proved.
\end{proof}

\begin{proof}[\bf Proof of Proposition~\ref{prop:bootdim2}.]
For $(s,t) \in [0,1]^2$ and $b=1,\dots,B$, let
$$
\bar \Db_n^{(b)}(s,t) = \lambda_n(s,1) \bar \Ab_n^{(b)} (0,s,t) - \lambda_n(0,s) \bar \Ab_n^{(b)} (s,1,t), 
$$
where $\bar \Ab_n^{(b)}$ is defined in~\eqref{eq:barAbnb}, and let
$$
\tilde \Db_n^{(b)}(s,t) = \lambda_n(s,1) \tilde \Ab_n^{(b)} (0,s,t) - \lambda_n(0,s) \tilde \Ab_n^{(b)} (s,1,t), 
$$
where $\tilde \Ab_n^{(b)}$ is defined in~\eqref{eq:tildeAbnb}. From Lemma~\ref{lem:aetildeAnb}, we then immediately obtain that
$$
\sup_{(s,t) \in [0,1]^2} | \tilde \Db_n^{(b)}(s,t) - \bar \Db_n^{(b)}(s,t) | = o_\Prob(1).
$$
The latter combined with Lemma~\ref{lem:multbarAn}, the continuous mapping theorem,~\eqref{eq:aeCbn},~\eqref{eq:sup} and~\eqref{eq:Dbn} gives
$$
\Big( \Db_n, \tilde \Db_n^{(1)},\dots, \tilde \Db_n^{(B)} \Big) \weak \Big( \Db_C,\Db_C^{(1)},\dots,\Db_C^{(B)} \Big)
$$
in $(\ell^\infty([0,1]^2), \| \cdot \|_\infty)^{B+1}$, where $\Db_C$ is defined in~\eqref{eq:DbC} and $\Db_C^{(1)},\dots,\Db_C^{(B)}$ are independent copies of $\Db_C$. From the definitions in~\eqref{eq:checkDnb} and~\eqref{eq:checkAbn}, we further have that, for $(s,t) \in [0,1]^2$, 
$$
\check \Db_n^{(b)}(s,t) 
= \lambda_n(s,1) \check \Ab_n^{(b)} (0,s,t) - \lambda_n(0,s) \check \Ab_n^{(b)} (s,1,t). 
$$
Hence, to complete the proof, it remains to show $\sup_{(s,t) \in [0,1]^2} |\check \Db_n^{(b)}(s,t) - \tilde \Db_n^{(b)}(s,t)| = o_\Prob(1)$, which is implied by
\begin{equation}
\label{eq:aeAbncheckbar}
\sup_{(s,s',t) \in \Delta \times [0,1]} |\check \Ab_n^{(b)}(s,s't) - \tilde \Ab_n^{(b)}(s,s',t)| = o_\Prob(1).
\end{equation}
Having in mind the fact that $A_{1:n}$ converges uniformly in probability to $A$ as a consequence of Proposition~\ref{prop:ferreira},~\eqref{eq:aebarBnb} and the fact $|\dot C_1| \leq 1$ and $|\dot C_2| \leq 1$, to prove~\eqref{eq:aeAbncheckbar}, it suffices to show that
\begin{equation} \label{eq:aepd1}
\sup_{(s,s') \in \Delta \atop (y,t) \in [0,1]^2} \left| \{ \dot C_{1,\ip{ns}+1:\ip{ns'},n}(y^{1-t},y^t) - \dot C_1(y^{1-t},y^t) \} \tilde \Bb^{(b)}_n(s,s',y^{1-t},1) \right| = o_\Prob(1)
\end{equation}
and
\begin{equation} \label{eq:aepd2}
\sup_{(s,s') \in \Delta \atop (y,t) \in [0,1]^2} \left| \{ \dot C_{2,\ip{ns}+1:\ip{ns'},n}(y^{1-t},y^t) - \dot C_2(y^{1-t},y^t) \} \tilde \Bb^{(b)}_n(s,s',1,y^t) \right| = o_\Prob(1),
\end{equation}
where $\dot C_{1,\ip{ns}+1:\ip{ns'},n}(y^{1-t},y^t)$, $\dot C_{2,\ip{ns}+1:\ip{ns'},n}(y^{1-t},y^t)$, $\dot C_1(y^{1-t},y^t)$ and $\dot C_2(y^{1-t},y^t)$ are given in~\eqref{eq:estdotc1},~\eqref{eq:estdotc2}, \eqref{eq:dotc1} and~\eqref{eq:dotc2}, respectively. 

From~\eqref{eq:boundCkl}, it can be verified that, for any $n \geq 1$ and any $t \in [0,1]$, $| \hat A_{1:n}(t) | \leq 7$. Since, by definition, $\hat A'_{k:\ell,n}$ in~\eqref{eq:hatAkl'} is also uniformly bounded, $(y,t) \mapsto \dot C_{1,k:\ell,n}(y^{1-t},y^t)$ and $(y,t) \mapsto \dot C_{2,k:\ell,n}(y^{1-t},y^t)$ are uniformly bounded. 

To prove~\eqref{eq:aepd1}, we can then proceed as in the proof of Proposition~4.3 of \cite{BucKojRohSeg14} (see the terms (B.4) and (B.5)). Using the asymptotic uniform equicontinuity in probability of $\tilde \Bb_n^{(b)}$ which vanishes when $s=s'$, and the fact that $\dot C_1$ and its estimator are uniformly bounded, it remains to show that, for any $\delta \in (0,1)$,  
\begin{equation} \label{eq:aepd1delta}
\sup_{(s,s',y,t) \in \Delta \times [0,1]^2 \atop s' - s \geq \delta} \left|\dot C_{1,\ip{ns}+1:\ip{ns'},n}(y^{1-t},y^t) - \dot C_1(y^{1-t},y^t)  \right| = o_\Prob(1).
\end{equation}
The previous result is implied by the fact that
$$
\sup_{(s,s',t) \in \Delta \times [0,1] \atop s' - s \geq \delta} \left| \hat A_{\ip{ns}+1:\ip{ns'}}(t) - A(t) \right| = \sup_{(s,s',t) \in \Delta \times [0,1] \atop s' - s \geq \delta} \left| \frac{\Ab_n(s,s',t)}{\sqrt{n} \lambda_n(s,s')} \right| = o_\Prob(1),
$$
where $\Ab_n$ is defined in~\eqref{eq:Abn}, and an analogue result for $\hat A_{\ip{ns}+1:\ip{ns'},n}'$ defined in~\eqref{eq:hatAkl'}. The latter can be seen as follows: for $A_{\ip{ns}+1:\ip{ns'},n}'$ in~\eqref{eq:Akl'}, we have
\begin{multline}
\label{eq:consistA'}
\sup_{(s,s',t) \in \Delta \times [h_n,1-h_n] \atop s' - s \geq \delta}  |A'_{\ip{ns}+1:\ip{ns'},n}(t) - A'(t)| \leq \sup_{t \in [h_n,1-h_n]} \left| \frac{A(t + h_n) - A(t - h_n)}{2 h_n} - A'(t)  \right| \\ + \sup_{(s,s',t) \in \Delta \times [h_n,1-h_n] \atop s' - s \geq \delta} \left| \frac{ \Ab_n(s,s',t + h_n) - \Ab_n(s,s',t - h_n) }{ 2 h_n \sqrt{n} \lambda_n(s,s') } \right|.
\end{multline}
Since $A'$, extended by continuity, is (uniformly) continuous on $[0,1]$ (see the discussion below~\eqref{eq:dotc2}), by the mean value theorem, the first term on the right converges to zero. The second term on the right is smaller than
$$
\sup_{(s,s',t,t') \in \Delta \times [0,1]^2 \atop | t - t' | \leq 2 h_n} \left| \Ab_n(s,s',t) - \Ab_n(s,s',t')  \right|  \times \sup_{(s,s') \in \Delta \atop s' - s \geq \delta} \frac{1}{ 2 h_n  \lambda_n(s,s') } = o_\Prob(1),
$$
by asymptotic uniform equicontinuity in probability of $\Ab_n$.

Hence,~\eqref{eq:aepd1delta} holds and, thus, so does~\eqref{eq:aepd1}. The proof of~\eqref{eq:aepd2} is similar. This completes the proof of~\eqref{eq:aeAbncheckbar} and, therefore, of the proposition.
\end{proof}

\begin{proof}[\bf Proof of Proposition~\ref{prop:bootdim2break}] For $b=1,\dots,B$, let $\check \Ab_n^{\theta,(b)}$ be the analogue of $\check \Ab_n^{(b)}$ in~\eqref{eq:checkAbn} based on the adapted pseudo-observations defined in~\eqref{eq:pseudotheta}. Then, by definition of $\check \Db_n^{\theta,(b)}$, we have 
$$
\check \Db_n^{\theta,(b)}(s,t) = \lambda_n(s,1) \check \Ab_n^{\theta,(b)}(0,s,t) - \lambda_n(0,s) \check \Ab_n^{\theta,(b)}(s,1,t), \quad (s,t) \in [0,1]^2.
$$
Next, let
\begin{multline*}
\tilde \Db_n^{\theta,(b)}(s,t) = \lambda_n(s,1) \{ \tilde \Ab_n^{(b)}(0,s \wedge \theta,t) + \tilde \Ab_n^{(b)}(s \wedge \theta, s,t) \} \\ - \lambda_n(0,s) \{ \tilde \Ab_n^{(b)}(s, s \vee \theta, t) + \tilde \Ab_n^{(b)}(s \vee \theta, 1, t) \}, \quad (s,t) \in [0,1]^2,
\end{multline*}
where $\tilde \Ab_n^{(b)}$ is defined in~\eqref{eq:tildeAbnb}, and let us first show that 
\begin{equation}
\label{eq:aetildeDbnbtheta}
\sup_{(s,t) \in [0,1]^2} | \check \Db_n^{\theta,(b)}(s,t) - \tilde \Db_n^{\theta,(b)}(s,t) | = o_\Prob(1).
\end{equation}
To prove~\eqref{eq:aetildeDbnbtheta}, we shall show that
\begin{multline}
\label{eq:aeDbn1}
\sup_{(s,t) \in [0,\theta] \times [0,1]} \left| \check \Db_n^{\theta,(b)}(s,t) - \lambda_n(s,1) \tilde \Ab_n^{(b)}(0,s,t) \right. \\ \left. + \lambda_n(0,s) \{ \tilde \Ab_n^{(b)}(s,\theta,t) + \tilde \Ab_n^{(b)}(\theta,1,t) \} \right| = o_\Prob(1),
\end{multline}
and 
\begin{multline}
\label{eq:aeDbn2}
\sup_{(s,t) \in [\theta,1] \times [0,1]} \left| \check \Db_n^{\theta,(b)}(s,t) - \lambda_n(s,1) \{ \tilde \Ab_n^{(b)}(0,\theta,t) + \tilde \Ab_n^{(b)}(\theta,s,t) \} \right. \\ \left. + \lambda_n(0,s) \tilde \Ab_n^{(b)}(s,1,t)\right| = o_\Prob(1).
\end{multline}
We start with the proof of~\eqref{eq:aeDbn1}. Under $\Hc_{0,c} \cap \Hc_{1,m}$, \eqref{eq:aeAbncheckbar} continues to hold if the supremum is restricted to  $(s,s',t)  \in (\Delta \cap[0,\theta]^2)  \times [0,1]$. This can be seen by essentially the same arguments as in the proof of Lemma~\ref{lem:aeCbnH1}. Therefore,~\eqref{eq:aeDbn1} will hold if
$$
\sup_{(s,t) \in [0,\theta] \times [0,1]} \left| \check \Ab_n^{\theta,(b)}(s,1,t) - \tilde \Ab_n^{(b)}(s,\theta,t) - \tilde \Ab_n^{(b)}(\theta,1,t)\right| = o_\Prob(1),
$$
that is, if 
\begin{equation}
\label{eq:convBntheta}
\sup_{s \in [0,\theta] \atop (u,v) \in [0,1]^2} \left| \check \Bb_n^{\theta,(b)}(s,1,u,v) - \tilde \Bb_n^{(b)}(s,\theta,u,v) - \tilde \Bb_n^{(b)}(\theta,1,u,v) \right| = o_\Prob(1),
\end{equation}
and, having in addition~\eqref{eq:dotc1},~\eqref{eq:dotc2},~\eqref{eq:estdotc1} and~\eqref{eq:estdotc2} in mind, if 
\begin{gather}
\label{eq:convAtheta}
\sup_{(s,t) \in [0,\theta] \times [0,1]} | \hat A_{\ip{ns}+1:n}^\theta(t) - A(t) | = o_\Prob(1), \\
\label{eq:convAprimetheta}
\sup_{(s,t) \in [0,\theta] \times [0,1]} | \hat A_{\ip{ns}+1:n,n}'^\theta(t) - A'(t) | = o_\Prob(1),
\end{gather}
where $\hat A_{k:\ell,n}'^\theta$ is the analogue of $\hat A_{k:\ell,n}'$ in~\eqref{eq:hatAkl'} defined from the adapted pseudo-observations in~\eqref{eq:pseudotheta}. The supremum on the left of~\eqref{eq:convBntheta} is smaller than
\begin{multline*}
\sup_{(s,u,v) \in [0,\theta] \times [0,1]^2} \left| \frac{\ip{n\theta} - \ip{ns}}{n - \ip{ns}} C_{\ip{ns}+1:\ip{n\theta}}(u,v) \right. \\ \left. + \frac{n - \ip{n\theta}}{n - \ip{ns}} C_{\ip{n\theta}+1:n}(u,v) - C(u,v) \right| \times \sup_{s \in [0,\theta]} \left| \frac{1}{\sqrt{n}} \sum_{i=\ip{ns}+1}^n \xi_i^{(b)} \right| \\
\leq \sup_{(s,u,v) \in [0,\theta] \times [0,1]^2} \frac{| \Cb_n(s,\theta,u,v) + \Cb_n(\theta,1,u,v) |}{\sqrt{n} \lambda_n(s,1)} \times O_\Prob(1) = o_\Prob(1)
\end{multline*}
by Lemma~\ref{lem:aeCbnH1} and the weak convergence of $\bar \Cb_n$ under $\Hc_{0,c}$. The supremum on the left of~\eqref{eq:convAtheta} is smaller than
$$
n^{-1/2} \sup_{(s,t) \in [0,\theta] \times [0,1]}  \{\lambda_n(s,1) \}^{-1}| \Ab_n^\theta(s,1,t)|,
$$
where $\Ab_n^\theta$ is defined in~\eqref{eq:Abntheta}, and is $o_\Prob(1)$ because of~\eqref{eq:ae1},~\eqref{eq:tildeAbn}, Lemma~\ref{lem:aeCbnH1} and the weak convergence of $\bar \Cb_n$ under $\Hc_{0,c}$. The proof of~\eqref{eq:convAprimetheta} is based on a decomposition similar to that used in~\eqref{eq:consistA'} and relies again on~\eqref{eq:ae1}. Hence,~\eqref{eq:aeDbn1} holds. The proof of~\eqref{eq:aeDbn2} is similar. 

Using arguments of the same nature as those employed in the proof of Lemma~\ref{lem:aeCbnH1}, we obtain the following extension of  Lemma~\ref{lem:aetildeAnb} under $\Hc_{1,m} \cap \Hc_{0,c}$:
$$
\sup_{(s,s',t) \in \{ (\Delta\cap[0,\theta]^2) \cup( \Delta\cap[\theta,1]^2)\} \times [0,1]} | \tilde \Ab_n^{(b)}(s,s',t) - \bar \Ab_n^{(b)}(s,s',t) | = o_\Prob(1),
$$
By the triangular inequality, this implies that
\begin{equation}
\label{eq:aebarDbnbtheta}
\sup_{(s,t) \in [0,1]^2} | \tilde \Db_n^{\theta,(b)}(s,t) - \bar \Db_n^{\theta,(b)}(s,t) | = o_\Prob(1),
\end{equation}
where
\begin{multline*}
\bar \Db_n^{\theta,(b)}(s,t) = \lambda_n(s,1) \{ \bar \Ab_n^{(b)}(0,s \wedge \theta,t) + \bar \Ab_n^{(b)}(s \wedge \theta, s,t) \} \\ - \lambda_n(0,s) \{ \bar \Ab_n^{(b)}(s, s \vee \theta, t) + \bar \Ab_n^{(b)}(s \vee \theta, 1, t) \}, \quad (s,t) \in [0,1]^2.
\end{multline*}
and $\bar \Ab_n^{(b)}$ is defined in~\eqref{eq:barAbnb}. The desired result is then a consequence of Lemma~\ref{lem:multbarAn}, the continuous mapping theorem,~\eqref{eq:aetildeDntheta},~\eqref{eq:aebarDntheta},~\eqref{eq:aetildeDbnbtheta} and~\eqref{eq:aebarDbnbtheta}.
\end{proof}


\section{Test statistic and multiplier bootstrap for $d\geq2$}
\def\theequation{B.\arabic{equation}}
\setcounter{equation}{0}
\label{sec:extensiond>2}

In Sections~\ref{sec:testdim2} and~\ref{sec:testdim2break}, we restricted ourselves to the case $d=2$. Results for arbitrary dimension $d\geq2$ can be established at the cost of a more complex notation but without significant additional mathematical difficulties. We give the main steps of the generalization hereafter. Let $\vect{X}=\left(X_{1},\ldots,X_{d}\right)$ be a random vector with c.d.f.\ and extreme-value copula of the form~\eqref{eq:sklar} and~\eqref{eq:evc}, respectively, and suppose that $A$ is continuously differentiable on the interior of $\Sc_{d-1}$ with partial derivatives $\dot{A}_j(\vect t)=\partial A(\vect{t})/\partial t_j$, $j=2,\ldots,d$. With the notation $U_j=F_j(X_j)$, $j=1,\dots,d$, and $t_1=t_1(\vect{t})=1-\sum\nolimits_{j=2}^dt_j$, $\vect t \in \Sc_{d-1}$, we have, just as for $d=2$, 
$$
A(\vect{t})=S(\vect{t})/ \{ 1-S(\vect{t}) \} \quad \mbox{and} \quad S(\vect{t})=\Exp \left( \max_{1\leq j\leq d}U_j^{1/t_j} \right),
$$ 
with the convention that $u^{1/0} = 0$ for all $u \in (0,1)$.

Let $\vect{X}_i$, $i=1,\ldots,n$, be independent copies of $\vect{X}$ and let $\hat{\vect{U}}_{k:\ell,i}=(\hat{U}_{k:\ell,i1},\ldots,\hat{U}_{k:\ell,id})$ be $d$-variate generalizations of the ``subsample'' pseudo-observations in~\eqref{eq:pseudoobs}. We define a CUSUM-type process $\Db_n$ on $[0,1]\times \Sc_{d-1}$ by
\begin{align*}
\Db_n(s,\vect{t})=\frac{\ip{ns}(n-\ip{ns})}{n^{3/2}}\left\{\hat{A}_{1:\ip{ns}}(\vect{t})-\hat{A}_{\ip{ns}+1:n}(\vect{t})\right\}, 
\end{align*}
where, for $1\le k \le \ell \le n$, $\hat A_{k:\ell}(\vect t) =\hat S_{k:\ell}(\vect t) / \{ 1- \hat S_{k:\ell}(\vect t) \}$, and
\[
\hat{S}_{k:\ell}(\vect{t})=\frac{1}{\ell-k+1}\sum_{i=k}^\ell\max_{1\leq j\leq d}\left(\hat{U}_{k:\ell,ij}^{1/t_j}\right)
\]
with the convention that $\hat S_{k:\ell} = 0$ if $k > \ell$.

Let us introduce some additional notation. For any $y\in[0,1]$ and $\vect{t}\in\Sc_{d-1}$, we define $\vect y^{\vect{t}}$ to be the vector $(y^{t_1},\ldots,y^{t_d})\in[0,1]^d$ with the convention that $0^0=1$. Furthermore, for any $\vect u \in [0,1]^d$ and any $j=1,\dots,d$, $\vect u^{(j)}$ denotes the vector of $[0,1]^d$ whose components are all equal to one, except the $j$th which is equal to $u_j$.

\begin{prop} \label{prop:weak}
Suppose that all of the above conditions are met. Then, in the normed space $\ell^\infty([0,1]\times\Sc_{d-1})$, $\Db_n \weak \Db_C$, where
\[
\Db_C(s,\vect{t})= \{ 1+ A(\vect t) \}^2 \times \int_0^1s \Cb_C(s,1,\vect y^{\vect{t}}) - (1-s) \Cb_C(0,s,\vect y^{\vect{t}}) \, dy.
\]
Here, $\Cb_C$ denotes a centered Gaussian process on $\Delta \times [0,1]^d$ defined through
\[
\Cb_C(s,s',\vect{u})= \{ \Bb_C(s',\vect{u}) - \Bb_C(s,\vect{u})  \} -\sum_{j=1}^d \dot C_j(\vect{u}) \{ \Bb_C(s',\vect{u}^{(j)}) - \Bb_C(s,\vect{u}^{(j)})  \},   
\]
where $\Bb_C$ is a tight, centered Gaussian process on $(\ell^\infty([0,1]^{d+1}), \| \cdot\|_\infty)$ with covariance kernel given by
\[
\Exp\{ \Bb_C(s,\vect{u})\Bb_C(s',\vect{u}') \} = (s\wedge s') \{ C(u_1\wedge u'_1,\ldots,u_d\wedge u_d') - C(\vect{u})C(\vect{u}') \},
\]
and $\dot C_j$, $j=1,\ldots,d$, denotes the $j$th first-order partial derivative of $C$.  
\end{prop}

The proof is almost identical to that of Proposition~\ref{prop:weakdim2}. For a corresponding bootstrap approximation of the limit $\Db_C$, let $\xi_i^{\scriptscriptstyle (b)}$, $i=1,\ldots,n,\ b=1,\ldots,B$, be i.i.d.\ standard normal multipliers. Furthermore, from~\eqref{eq:evc}, we have that, for any $y\in(0,1)$ and $\vect t \in \Sc_{d-1}$,
\begin{align*}
\dot{C}_j(\vect y^{\vect{t}})=\left\{\begin{array}{ll}
\vect y^{A(\vect{t})-t_1} \left\{ A(\vect{t})-\sum\nolimits_{j'=2}^dt_{j'} \dot{A}_{j'}(\vect{t})\right\}, & j=1, \\
\vect y^{A(\vect{t})-t_j}  \left\{ A(\vect{t})+\dot A_j(\vect{t})-\sum\nolimits_{j'=2}^dt_{j'} \dot A_{j'}(\vect{t})\right\}, & j=2,\ldots,d.
\end{array}\right.
\end{align*}
The above quantities can be estimated consistently by plugging in subsample estimators of $A$ and $\dot A_j$, $j=1,\dots,d$, respectively, namely $\hat A_{k:\ell}$ and 
\begin{align*}
\dot{A}_{j,k:\ell,n}(\vect{t})=\frac{1}{2h_n}\left\{\hat A_{k:\ell}(\vect{t}+h_n\vect{e}_j)-\hat A_{k:\ell}(\vect{t}-h_n\vect{e}_j)\right\},\ \text{ }\ j=2,\ldots,d,
\end{align*}
with $\vect{t} \pm h_n\vect{e_j}=(t_2,\ldots,t_{j-1},t_j \pm h_n,t_{j+1},\ldots,t_d)$ and a sequence $h_n\downarrow0$ such that $\inf_{n \geq 1} h_n \sqrt{n}>0$ (boundary effects can be dealt with by generalizing the approach adopted below~\eqref{eq:Akl'}). Then, analogously to the bivariate case, we define
$$
\check{\Db}_n^{(b)}(s,\vect{t})
= 
\{ 1+ \hat A_{1:n}(\vect t) \}^2 \times 
\bigg\{ 
\frac{\ip{ns}}{n^{3/2}} \sum_{i=\ip{ns}+1}^n\xi_i^{(b)} \hat{w}_{\ip{ns}+1:n,i}(\vect{t}) -\frac{n-\ip{ns}}{n^{3/2}}\sum_{i=1}^{\ip{ns}}\xi_i^{(b)} \hat{w}_{1:\ip{ns},}i(\vect{t})
\bigg\},
$$
where, for $1\le k\le \ell \le n$,
\begin{align*} 
\hat{w}_{k:\ell,i}(\vect{t})
=
\widebar{m}_{k:\ell}(\vect{t})-\hat{m}_{k:\ell,i}(\vect{t})
+
\sum_{j=1}^d\frac{(\hat{u}_{k:\ell, ij}(\vect{t})-\widebar{u}_{k:\ell,j}(\vect{t}))\hat{a}_{k:\ell,j}(\vect{t})}{\hat{b}_{k:\ell,j}(\vect{t})},
\end{align*}
with $\bar{m}_{k:\ell}$ and $\bar{u}_{k:\ell,j}$ denoting the arithmetic mean over $i=k, \dots, \ell$ of
\begin{align*}
\hat{m}_{k:\ell,i}(\vect{t})=\max\left(\hat{\vect{U}}_{k:\ell,i}^{1/\vect{t}}\right)\ 
\text{ and }\ 
\hat{u}_{k:\ell,ij}(\vect{t})=\hat{U}_{k:\ell,ij}^{\hat{b}_{k:\ell,j}/t_j},
\end{align*}
and where
\begin{align*}
\hat{a}_{k:\ell,j}(\vect{t}) &= 
\begin{cases}
\hat A_{k:\ell}(\vect{t})-\sum\nolimits_{j'=2}^d t_{j'}\dot{A}_{j',k:\ell,n}(\vect{t}), & j=1, \\
\hat A_{k:\ell}(\vect{t})+\dot{A}_{j,k:\ell,n}(\vect{t})-\sum\nolimits_{j'=2}^dt_{j'}\dot{A}_{j',k:\ell,n}(\vect{t}), & j=2,\ldots,d,
\end{cases} \\
\hat{b}_{k:\ell,j}(\vect{t}) &= \hat A_{k:\ell}(\vect{t})+1-t_j.
\end{align*}

Test statistics and corresponding multiplier bootstrap replicates can be defined analogously to Section~\ref{sec:testdim2}, as functionals of $\Db_n$ and $\check \Db_n^{\scriptscriptstyle (b)}$, $b=1,\dots,B$, respectively. In addition, generalizations adapted to known breaks in the margins can be obtained by computing pseudo-observations from the subsamples determined by the marginal change-points, as explained in Section~\ref{sec:testdim2break}. We omit the details for the sake of brevity.

\bigskip
\noindent
\textbf{Acknowledgements.} 
The authors would like to thank Markus Schulte and Andreas Schumann for providing us with the hydrological data sets and Betina Berghaus and Roland Fried for helpful discussions.
This work has been supported by the Collaborative Research Center ``Statistical modeling of nonlinear dynamic processes'' (SFB 823) of the German Research Foundation (DFG) which is gratefully acknowledged. 

\bibliographystyle{chicago}
\bibliography{biblio}
\end{document}